\documentclass[acmsmall,nonacm]{acmart}

\newif\ifarxiv
\arxivtrue

\usepackage{amsthm, amsmath, mathtools}
\usepackage{upgreek}
\usepackage{wrapfig}
\usepackage{amsfonts}
\usepackage{enumitem}
\usepackage[normalem]{ulem}
\usepackage[linesnumbered,vlined,ruled,commentsnumbered]{algorithm2e}
\usepackage[para,online,flushleft]{threeparttable}
\usepackage{multirow}
\usepackage{multicol}
\usepackage{hyperref}
\usepackage{url}
\usepackage{bm}
\usepackage[labelfont=bf,skip=0pt]{caption,subfig}
\usepackage{flushend}
\usepackage{todonotes}

\usepackage{tabularx}
\usepackage{bbm}

\usepackage[english]{babel}
\usepackage{graphicx}
\graphicspath{{./figs/}}
\usepackage{thmtools}

\usepackage{listings}
\usepackage{tablefootnote}
\usepackage{enumitem}

\usepackage{bbm}

\usepackage{multirow}

\DeclareMathOperator*{\argmin}{arg\,min}

\newcommand{\opt}{\textsf{opt}}

\newcommand{\Q}{\mathcal{Q}}

\newcommand{\T}{\mathcal{T}}

\newcommand{\eps}{\varepsilon}
\renewcommand{\epsilon}{\varepsilon}
\renewcommand{\Re}{\mathbb{R}}

\newcommand{\allattr}{{\mathbf A}}

\newcommand{\allrel}{{\mathbf R}}

\newcommand{\dom}{{\texttt{dom}}}

\newcommand{\poly}{{\textsf{poly}}}

\newcommand{\fhw}{}
\newcommand{\query}{\mathcal{Q}}
\newcommand{\I}{\mathbf{I}}

\newcommand{\dist}{\mathsf{dist}}

\newcommand{\probvector}{\sigma}

\newcommand{\canonical}{\mathcal{C}}
\newcommand{\tree}{\mathcal{T}}
\newcommand{\rangetree}{\mathcal{G}}

\newcommand{\stavros}[1]{{\color{red}Stavros says: #1}}

\newcounter{lpcounter}
\newcommand{\lpproblemtext}{%
  \makebox[0pt][r]{\text{($\mathsf{LP}$\arabic{lpcounter})}\quad}%
}

\newcommand{\lpproblemlabel}[1]{%
  \refstepcounter{lpcounter}%
  \label{lp:#1}%
}

\newcounter{fpcounter}

\newcommand{\prob}{\mathsf{CSO}}
\newcommand{\gprob}{\mathsf{GCSO}}
\newcommand{\relclusterone}{\mathsf{RCRO}}
\newcommand{\relclustertwo}{\mathsf{RCTO}}

\newcommand{\expon}{2.38}
\newcommand{\setrects}{\mathcal{R}}
\renewcommand{\O}{\tilde{O}}

\begin{document}

\lpproblemlabel{1}

\title{Clustering with Set Outliers and Applications in Relational Clustering}

\author{Vaishali Surianarayanan}
\orcid{0000-0003-3091-38238}
\affiliation{%
  \institution{
  University of California, Santa Cruz}
  \city{Santa Cruz}
  \country{USA}
}
\email{vaishalisurianarayanan@gmail.com}

\author{Neeraj Kumar}
\orcid{0000-0001-9356-526X}
\affiliation{%
  \institution{Meta Platforms}
  \city{Menlo Park}
  \country{USA}
}
\email{17.neeraj.s@gmail.com}

\author{Stavros Sintos}
\orcid{0000-0002-2114-8886}
\affiliation{%
  \institution{
  University of Illinois Chicago}
  \city{Chicago}
  \country{USA}
}
\email{stavros@uic.edu}

\renewcommand{\shortauthors}{Vaishali Surianarayanan, Neeraj Kumar, \& Stavros Sintos}

\ccsdesc[500]{Theory of computation~Database theory}
\ccsdesc[500]{Theory of computation~Data structures and algorithms for data management}

\keywords{clustering, $k$-center, outliers, relational clustering, relational data}

\begin{abstract}
We introduce and study the $k$-center clustering problem with \emph{set outliers}, a natural and practical generalization of the classical $k$-center clustering with outliers. Instead of removing individual data points, our model allows discarding up to $z$ subsets from a given family of candidate outlier sets $\mathcal{H}$. More formally, given a metric space $(P,\dist)$, where $P$ is a set of elements and $\dist$ a distance metric, a family of sets $\mathcal{H}\subseteq 2^P$, and parameters $k, z$, the goal is to compute a set of $k$ centers $S\subseteq P$ and a family of $z$ sets $H\subseteq \mathcal{H}$ to minimize $\max_{p\in P\setminus(\bigcup_{h\in H} h)} \min_{s\in S}\dist(p,s)$ (clustering cost). This abstraction captures structured noise common in database applications, such as faulty data sources or corrupted records in data integration and sensor systems. 

We present the first approximation algorithms for this problem in both general and geometric settings. Our methods provide tri-criteria approximations: selecting up to $2k$ centers and $2f z$ outlier sets (where $f$ is the maximum number of sets that a point belongs to), while achieving constant-factor approximation in clustering cost. In geometric settings, we leverage range and BBD trees to achieve near-linear time algorithms. In many real applications $f=1$. In this case we further improve the running time of our algorithms by constructing small \emph{coresets}. We also provide a hardness result for the general problem showing that it is unlikely to get any sublinear approximation on the clustering cost selecting less than $f\cdot z$ outlier sets.

We demonstrate that this model naturally captures relational clustering with outliers. We define and study two new formulations: one where outliers are result tuples in a join, and another where outliers are input tuples whose removal affects the join output. We provide approximation algorithms for both, establishing a tight connection between robust clustering and relational query evaluation.
\end{abstract}

\maketitle
\section{Introduction}
\label{sec:intro}

Clustering is a fundamental technique in computer science, widely applied in domains such as social network analysis, search result grouping, medical imaging, and image segmentation. By partitioning data into coherent groups, clustering helps uncover hidden structures and relationships that inform downstream tasks like classification, anomaly detection, and recommendation.

However, real-world data is often noisy due to human error or data integration issues, and many classical clustering objectives are not robust to such noise. For instance, in $k$-center clustering, where the goal is to minimize the maximum distance of any point to its closest center, a few outlier points can drastically alter the resulting clusters, potentially leading to misleading interpretations.
To address this, prior work has studied robust versions of clustering, notably $k$-center clustering with $z$ outliers. Given a metric space $(P,\dist)$, where $P$ is a set of elements and $\dist$ a distance metric, the objective is to select $k$ centers and remove $z$ outliers so as to minimize the maximum distance from any remaining point to its closest center. This problem is known to be $\mathsf{NP}$-complete, but efficient constant-factor approximation algorithms exist~\cite{charikar1999constant, charikar2003better}.

In this work, we introduce a novel and natural generalization of this problem: \emph{$k$-center clustering with $z$ set outliers}. Instead of removing individual outlier points, the goal is to remove $z$ subsets from a given family $\mathcal{H} \subseteq 2^P$ of candidate outlier sets. This formulation captures practical scenarios in data cleaning, fault-tolerant systems, and federated data analysis. For example, each subset might represent data from a single source in a data integration task, or a group of measurements from a single sensor in a sensor network. The goal is to discard entire sources that may be faulty or unreliable, rather than identifying individual points.

Consider the following example showing an instance of our problem in the geometric case where sets are represented as hyper-rectangles.
Transactions or user activities are embedded as points with features like (price, quantity, time, store location). 
Hyper-rectangles naturally model ranges of values, e.g., Price range: $[10,20]$, Time window: $[1\textsf{pm},3\textsf{pm}]$. So a rectangle corresponds to all transactions in a spatio-temporal range. Using an ML classifier a data analyst can identify a set of potential fraudulent patterns or promotion windows represented as hyper-rectangles. For example, all transactions with price $[0,1]$ that took place between $[1\textsf{am}-2\textsf{am}]$ are likely spam or test transactions. 
The goal is to exclude up to $z$ suspicious rectangles to reveal meaningful transaction patterns/clusters. 

We present the first approximation algorithms for this problem, which we formalize under multiple settings: general metric spaces with arbitrary subsets, and geometric settings where points lie in $\mathbb{R}^d$ and outlier sets correspond to axis-aligned hyper-rectangles. Our algorithms offer tri-criteria approximations: we return at most $2k$ centers and up to $2fz$ outlier sets from $\mathcal{H}$, while achieving a constant-factor approximation in clustering cost, where $f$ is the maximum number of sets an element belongs to. While $f$ can be large in theory, it is often small in practical settings.
We also show that if the unique game conjecture is true, there is no hope to design a polynomial time algorithm that returns $k$ centers, $(f-\zeta)z$ outlier sets, having any sublinear approximation on the clustering cost, where $\zeta>0$ is any positive number.
Our algorithms run in near-quadratic time for general metrics and near-linear time in the geometric setting. These results offer robust and scalable methods for clustering with structured noise.

Finally, we show how this model can capture relational clustering problems. In relational databases, data is stored across multiple tables, and complete information is only revealed through joins.
This is the standard formulation in real database systems ~\cite{link1, link:relational}.
We define two versions of relational $k$-center clustering with outliers: one where outliers are tuples in the join result, and another where outliers are input tuples whose removal affects the join. The former reduces to the relatively well-studied problem of $k$-center clustering with outliers 
whereas the latter reduces to geometric $k$-center clustering with set outliers problem introduced in this paper.
We provide new approximation algorithms for both these variations, thus establishing a bridge between robust clustering and relational data management.

To further motivate relational clustering with outliers, consider a crowdsourcing scenario where data is aggregated from multiple online sources and stored in a relational database.
A data analyst may want to cluster the results of a join query to construct a classifier (unsupervised learning). However, even a handful of erroneous tuples (outliers) in one table can lead to arbitrarily large errors in k-center clustering on the join result. This problem is particularly relevant in crowdsourcing, where identifying and mitigating erroneous data via cleaning techniques is a central challenge~\cite{chu2015katara}. Efficient algorithms for k-center clustering with outliers in the relational setting are therefore essential.

\vspace{-0.5em}
\subsection{Notation and Problem Definition}
Let $(P,\dist)$ be a metric space where $P$ is a set of elements, $\dist$ is a metric function and $\dist(p,q)$ denotes the distance between elements $p,q\in P$.  
For a set of elements $P$ and a subset of centers $C\subseteq P$, let $\rho(C,P)=\max_{p\in P}\min_{c\in C}\dist(p,c)$ be the \emph{clustering cost} defined as maximum distance of an element in $P$ from its closest center in $C$. The \emph{$k$-center clustering} problem is a well-studied problem where given $(P$, $\dist)$, and a parameter $k$ the goal is to compute a set $C\subseteq P$ with $|C|\leq k$ such that the cost $\rho(C,P)$ is minimized. 
Prior work has also studied the \emph{$k$-center clustering with outliers} problem where given $P$, $\dist$, and two parameters $k, z$ the goal is to compute a set of centers $C\subseteq P$ with $|C|\leq k$ and a set of outliers $H\subseteq P$ with $|H|\leq z$, such that 
$\rho(C,P\setminus H)$ is minimized.
Let $\rho^*_{k,z}(P)$ be the cost of an optimal $k$-center clustering of $P$ with $z$ outliers. 

In the paper, we introduce some natural extensions of
$k$-center clustering with outliers problem. We first define the \emph{$k$-center Clustering problem with Set Outliers}, or $\prob$ in short, assuming that $\dist$ is any general metric function.
\begin{definition}[$\prob(P,\mathcal{H},k,z)$]
    Given a set $P$ of $n$ elements, a family $\mathcal{H}$ of $m$ subsets of $P$, where
    every element in $P$ belongs to at least one set in $\mathcal{H}$, and two integer positive parameters $k, z$, the goal is to compute a set of centers $C\subseteq P$ with $|C|\leq k$ and a family $H\subseteq \mathcal{H}$ of outlier sets with $|H|\leq z$ and $C\cap (\bigcup_{h\in H}h)=\emptyset$ such that the cost excluding the points covered by outliers $\rho(C,P\setminus (\bigcup_{h\in H} h))$ is minimized.
\end{definition}
Let $\rho^*_{k,z}(P,\mathcal{H})$ be the cost of an optimal solution of $\prob(P,\mathcal{H},k,z)$.
Throughout the paper we use the notation $f$ to denote the maximum number of sets that a point belongs to, i.e., $f=\max_{p\in P}|\{h\in \mathcal{H}\mid p\in H\}|$.
When we say that a valid solution $C, H$ is computed for $\prob(P,\mathcal{H},k,z)$ we mean that $C\subseteq P$, $H\subseteq \mathcal{H}$ and, $C\cap (\bigcup_{h\in H}h)=\emptyset$.
An algorithm returns a $(\mu_1, \mu_2, \mu_3)$-approximation for the $\prob(P,\mathcal{H},k,z)$ problem if it returns a valid solution $C, H$ with $|C|\leq \mu_1\cdot k$, $|H|\leq \mu_2\cdot z$, and $\rho(C,P\setminus(\bigcup_{h\in H}h))\leq \mu_3\cdot \rho^*_{k,z}(P,\mathcal{H})$.

We consider different versions of the $\prob$ problem. Recall that in the most general one, $P$ is a set of elements in a general metric space and every set in $\mathcal{H}$ is an arbitrary subset of $P$. We distinguish into cases where every element in $P$ lies in one or multiple sets in $\mathcal{H}$, i.e. $f=1$ or $f>1$. Furthermore, we consider the geometric case where $P$ is a set of points in $\Re^d$, $\dist$ is the Euclidean distance, and $\mathcal{H}$ is a set of hyper-rectangles in $\Re^d$. In this case every rectangle $h\in H$ defines a subset $h\cap P$. More formally, we define the \emph{Geometric $k$-center Clustering with rectangular Outliers}, or $\gprob$ in short.
\newcommand{\rects}{R}
\newcommand{\rec}{\square}
\begin{definition}[$\gprob(P,\setrects,k,z)$]
    Given a set $P$ of $n$ points in $\Re^d$, a set $\setrects$ of $m$ hyper-rectangles in $\Re^d$, where $d$ is a constant and every point in $P$ lies in at least one hyper-rectangle in $\setrects$, and two integer positive parameters $k, z$, the goal is to compute a set of centers $C\subseteq P$ with $|C|\leq k$ and a set $\rects\subseteq \setrects$ of hyper-rectangles as outliers with $|\rects|\leq z$ and  $C\cap (\bigcup_{\rec\in \rects}\rec\cap P)=\emptyset$, such that the cost excluding points covered by outliers $\rho(C,P\setminus (\bigcup_{\rec\in \rects} \rec\cap P))$ is minimized.
\end{definition}
Let $\rho^*_{k,z}(P,\setrects)$ be the cost of an optimum solution of $\gprob(P,\mathcal{H},k,z)$.
When we say that a valid solution $C, \rects$ is computed for $\gprob(P,\setrects,k,z)$ we mean that $C\subseteq P$, $R\subseteq \mathcal{R}$ and, $C\cap (\bigcup_{\rec\in \rects}\rec \cap P)=\emptyset$.
An algorithm returns a $(\mu_1, \mu_2, \mu_3)$-approximation for the $\gprob(P,\setrects,k,z)$ problem if it returns a valid solution $C, \rects$ with $|C|\leq \mu_1\cdot k$, $|\rects|\leq \mu_2\cdot z$, and $\rho(C,P\setminus(\bigcup_{\rec\in \rects}\rec\cap P))\leq \mu_3\cdot \rho^*_{k,z}(P,\setrects)$.

If $k,z$ are clear from the context, we write $\prob(P,\mathcal{H})$, $\gprob(P,\setrects)$.
Interestingly, as we will see later the $\gprob$ problem has a strong connection with the relational $k$-center problem with outliers. 
Next, we give the basic definitions in the relational setting, following~\cite{esmailpour2024improved}, and then give the formal definition of the relational $k$-center clustering with outliers.

\vspace{0.5em}
\noindent{\bf Join Queries.}
 Let $\allrel$ denote a database schema with $g$ relations $R_1,\ldots, R_g$ and let $\allattr$ denote the set of $d$ attributes. Each relation $R_i$ has a subset of attributes $\allattr_i \subseteq \allattr$.
 Let $\dom(A)$ denote the domain of attribute $A \in \allrel$.
We assume that all attributes have the domain $\Re$ of reals.
 A database instance $\I$ consists of the set $\{R_i^\I\}$ of relational instances, where each $R_i^\I$ is a set of tuples over the domain $\Re^{|\allattr_i|}$.
 When the context is clear, we will drop the superscript $\I$, and simply refer to relation instance $R_i^\I$ as relation $R_i$. Let $|\I|=N$. We study all our problems in the data complexity, i.e., $g, d=O(1)$, while $N$ is large.
For a set of tuples $X$ over a subset of attributes $A_X\subseteq\allattr$ and a set of attributes $Y\subseteq \allattr$, the operator $\pi_Y(X)$ returns the projection of every tuple $t\in X$ on the attributes $Y\cap A_X$. Notice that the projection is defined as a set so $\pi_Y(X)$ does not contain duplicates.
We consider \emph{join queries}:
     $\query := R_1 \Join  R_2 \Join \cdots \Join  R_g$.
 The result of $\query$ over database instance $\I$ is defined as
 $\query(\I) = \bigl\{
     t \in \Re^d \bigm| \forall R_i\in \allrel, \pi_{\allattr_i}(t)\in R_i
\bigr\}$.
Note that each tuple in $\query(\I)$ is essentially a point in $\Re^{d}$ and it might be the case that $|\Q(\I)|=O(N^{\mathbbm{p} (\Q)})\gg N$, where $\mathbbm{p}(\Q)$ is the fractional edge cover of $\Q$~\cite{atserias2013size}.
 A join query is acyclic \cite{beeri1983desirability, fagin1983degrees} if there exists a tree $\T$, called a {\em join tree} of $\Q$, where
the nodes of $\T$ are $R_1, \ldots, R_g$, and for each attribute $A \in \cup_{i\in[m]} \allattr_i$, the set of nodes whose attributes contain $A$ form an edge-connected subtree of $\T$.


\paragraph{Relational $k$-center clustering with outliers}
The standard relational $k$-center clustering (without outliers) has been defined in~\cite{agarwal2024computing, chen2022coresets}.
Given a database instance $\I$, a join query $\Q$, and a parameter $k$ the goal is to compute a set $C\subseteq \Q(\I)$ with $|C|\leq k$ such that $\rho(C,\Q(\I))$ is minimized.
In this paper, we define the problem of relational $k$-center clustering with outliers. There are two meaningful and practical definitions of this problem. In the first one outliers are tuples in the join results, while in the second one outliers are tuples from the input database instance.
For every problem in the relational setting we assume that $\dist$ is the Euclidean distance.

We define the \emph{$k$-Center Clustering with Tuple-Outliers} problem, or $\relclustertwo$ in short. In this case, outliers are tuples from the input database $\I$.
\begin{definition}[$\relclustertwo(\Q,\I, k, z)$]
Given a database schema $\allrel$, a database instance $\I$ of size $O(N)$, a join query $\Q$, and two integer positive parameters $k, z$, the goal is to compute a set $S \subseteq \Q(\I)$ with $|S|\leq k$ and a set $T\subseteq \I$ with $|T|\leq z$ and $S\subseteq \Q(\I\setminus T)$, such that $\rho(S,\Q(\I\setminus T))$ is minimized.
\end{definition}
\vspace{-0.7em}
Let $\hat{\rho}^*_{k,z}(\Q(\I))$ denote the cost of the optimum solution for the $\relclustertwo$ problem.
When we say that a valid solution $S, T$ is computed for $\relclustertwo(\Q,\I,k,z)$ we mean that $S\subseteq\Q(\I)$, $T\subseteq \I$ and $S\subseteq \Q(\I\setminus T)$.
An algorithm returns a $(\mu_1, \mu_2, \mu_3)$-approximation for the $\relclustertwo(\Q,\I,k,z)$ problem if it returns a valid solution $S, T$ with $|S|\leq \mu_1\cdot k$, $|T|\leq \mu_2\cdot z$, and $\rho(S,\Q(\I\setminus T))\leq \mu_3\cdot \hat{\rho}^*_{k,z}(\Q(\I))$.

\newcommand{\relclusterthree}{\mathsf{RCTO1}}
We also study a variation of $\relclustertwo(\Q,\I,k,z)$ where we only allow outliers from the tuples of one specific relation. Without loss of generality, assume that $T\subseteq R_1$.
This is common in practice, as we might have collected data in $R_1$ from an untrustworthy source, while the remaining tables in $\allrel$ come from trustworthy sources.
Let $\relclusterthree(\Q,\I,k,z)$ be the $\relclustertwo(\Q,\I,k,z)$ problem where $T$ should be a subset of $R_1$.
Let
$\hat{\rho}^{*}_{k,z,1}(\Q(\I))$ denote the cost of the optimum solution for the $\relclusterthree$ problem.
When we say that a valid solution $S, T$ is computed for $\relclusterthree(\Q,\I,k,z)$ we mean that $S\subseteq\Q(\I)$, $T\subseteq R_1$ and $S\subseteq \Q(\I\setminus T)$.
An algorithm returns a $(\mu_1, \mu_2, \mu_3)$-approximation for the $\relclusterthree(\Q,\I,k,z)$ problem if it returns a valid solution $S, T$ with $|S|\leq \mu_1\cdot k$, $|T|\leq \mu_2\cdot z$, and $\rho(S,\Q(\I\setminus T))\leq \mu_3\cdot \hat{\rho}^*_{k,z,1}(\Q(\I))$.

As we show later, an instance of the $\relclusterthree$ problem can be mapped to an instance of the $\gprob$ problem.
This connection is conceptually interesting: a problem defined as clustering with tuple outliers in the relational setting can be naturally reformulated as a clustering problem with set outliers in the standard computational setting.

Finally, we define the \emph{$k$-Center Clustering with Result-Outliers problem}, or $\relclusterone$ in short. In this case the outliers are tuples from the join results.
\begin{definition}[$\relclusterone(\Q,\I, k, z)$]
Given a database schema $\allrel$, a database instance $\I$ of size $O(N)$, a join query $\Q$, and two integer positive parameters $k, z$, the goal is to compute a set $S \subseteq \Q(\I)$ with $|S|\leq k$ and a set $T\subseteq \Q(\I)$ with $|T|\leq z$ and $S\cap T=\emptyset$, such that $\rho(S,\Q(\I)\setminus T)$ is minimized.
\end{definition}
The cost of the optimum solution of the $\relclusterone$ problem is $\rho^*_{k,z}(\Q(\I))$.
When we say that a valid solution $S, T$ is computed for $\relclusterone(\Q,\I,k,z)$ we mean that $S\subseteq\Q(\I)$, $T\subseteq \Q(\I)$ and $S\cap T=\emptyset$.
An algorithm returns a $(\mu_1, \mu_2, \mu_3)$-approximation for the $\relclusterone(\Q,\I,k,z)$ problem if it returns a valid solution $S, T$ with $|S|\leq \mu_1\cdot k$, $|T|\leq \mu_2\cdot z$, and $\rho(S,\Q(\I)\setminus T)\leq \mu_3\cdot \rho^*_{k,z}(\Q(\I))$.

\subsection{Related work}

Clustering is a fundamental problem in computer science with applications in various domains such as community detection, recommendation systems, and dimensionality reduction~\cite{rokach2005clustering}.
There are multiple clustering objective functions such as $k$-center (minimize the maximum distance from a center) $k$-median (minimize the sum of distances from every point to its closest center), and $k$-means (minimize the sum of square distances from every point to its closest center). While the problems are $\mathsf{NP}$-hard even in the geometric setting (in $2$-dimensions) there are efficient constant approximation factors that run in polynomial time~\cite{har2003coresets, charikar1999constant, krishnaswamy2018constant}.

Clustering with outliers is a well studied problem~\cite{chen2008constant, agrawal2023clustering}. For $k$-center clustering with outliers Charikar et al.~\cite{charikar2001algorithms} described a $3$-approximation algorithm that runs in polynomial time. Later, Charikar et al.~\cite{charikar2003better} showed a sampling-based method for the $k$-center clustering with outliers that runs in the streaming setting. 
Interestingly, they also show that running the more expensive algorithm from~\cite{charikar2001algorithms} on a sufficiently large number of samples from the input set, returns a constant approximation to the $k$-center clustering with outliers.
The problem of $k$-center clustering with outliers has been studied under various settings and models, such as MPC, streaming, distributed, dynamic, sliding window, and fairness models~\cite{de2023k, malkomes2015fast, ceccarello2019solving, chan2022fully, de2021k, amagata2024fair}.
While clustering with outliers is well-studied, it is unclear how to extend these algorithms to $k$-center clustering with set outliers. 

Recently, clustering problems have been studied in the relational setting. 
Curtin et al.~\cite{curtin2020rk} gave a $O(1)$-approximation algorithm for the relational $k$-means clustering problem on acyclic join queries in $\O(k^m\cdot N+\mathsf{Alg}_M(k^m))$ time, where
$\mathsf{Alg}_M(k^m)$ is the running time of a known $O(1)$-approximation $k$-means algorithm in the standard computational setting\footnote{Similarly to~\cite{curtin2020rk, moseley2021relational}, in the standard computational setting we assume that all data lie in one table.}. Moseley et al.~\cite{moseley2021relational} improved the running time to $\O(k^4N+\mathsf{Alg}_M(k))$. Esmailpour et al.~\cite{esmailpour2024improved} developed a $O(1)$-approximation for the relational $k$-means clustering in $\O(k^2N + k^4 + \mathsf{Alg}_M(k^2))$ time which also works for the relational $k$-median problem. For the relational $k$-center clustering, Agarwal et al.~\cite{agarwal2024computing} proposed a $(2+\eps)$-approximation algorithm in $\O(\min\{k^2N,kN+k^{d/2}\})$ time. Finally, Chen et al.~\cite{chen2022coresets} proposed additive approximation algorithms for different versions of relational clustering.

\subsection{Our contributions}
In this paper, we introduce and make non-trivial progress on several variations of \emph{$k$-center clustering with set outliers} which is a natural generalization of the well-known $k$-center clustering with outliers problem. Our contributions can be organized as follows. Please also check Table~\ref{tab:3col8row}.

\begin{itemize}[leftmargin=*]
    \item \textbf{Near Optimal Approximation for Low Frequency Set Systems.} For the most general $\prob(P,\mathcal{H})$ problem, we design an $\mathsf{LP}$-based $(2,2f,2)$-approximation algorithm (Section~\ref{subsec:GeneralIS}) that runs in $O((n+m)^{\expon}\log n)$ time. Our algorithm first solves a linear program which takes $O((n+m)^{\expon})$ time using~\cite{jiang2021faster} and then rounds its solution. We also show that our approximation is near optimal for low frequency set systems. In particular, we show that assuming Unique games conjecture~\cite{khot2002power, klarreich2011approximately} is true, there exists no $(1,f-\zeta,\gamma)$-approximation algorithm for $\prob$ in polynomial time for all $f = o(\log n)$. This shows that our approximation bounds are near optimal when each element belongs to $o(\log n)$ outlier sets (Section~\ref{sec:hardness}).
    
    \item \textbf{Faster Algorithms for Practical Use Cases.} Motivated by practical applications in data cleaning and relational clustering, we study several specializations of $\prob(P,\mathcal{H})$ problem, namely when the outliers sets are disjoint ($f=1$) and the geometric version $\gprob$ defined before. In both these cases, the idea is to use techniques such as coresets and geometric structures to expedite the $\mathsf{LP}$-solving step of $\prob(P,\mathcal{H})$. Specifically:

\begin{table}[t!]
\vspace{-0.5em}
\centering
\begin{tabular}{|c|c|c|c|}
\hline
\textbf{Problem} & \textbf{Approximation} & \textbf{Running time} & \textbf{Section}\\
\hline
$\prob$ (Lower bound) & $(1,f-\zeta,\gamma)$ & $\omega(\poly(n,m))$ & \ref{sec:hardness}\\
\hline
$\prob$, $f>1$ & $(2,2f,2)$ & $\O((n+m)^{\expon})$ & \ref{subsec:GeneralIS}\\
\hline
$\prob$, $f=1$ & $(2,2,O(1))$ & $\O(nk+k^2m+\beta_1^2+(\beta_2+\beta_3)^{\expon})$ & \ref{subsec:GeneralDS}\\
\hline
$\gprob$, $f>1$ & $(2+\eps,2f,2+\eps)$ & $\O((k+z)(n+m))$ & \ref{subsec:GeometricIS}\\
\hline
$\gprob$, $f=1$ & $(2+\eps,2,O(1))$ & $\O(n+km+(k+z)(\beta_2'+\beta_3))$ & \ref{subsec:GeometricDS}\\
\hline
$\relclusterthree$ & $(2+\eps,2,O(1))$ & $\O(k^2N^2)$ & \ref{subsubsec:TupleOutliersOneRel}\\
\hline
$\relclustertwo$ & $(1,g,O(1))$ & $\O(2^{g\cdot k+z}(z+k^d)N)$& \ref{subsubsec:TupleOutliersMultRels} \\
\hline
$\relclusterone$ & $(1,1+\eps,3+\eps)$ & $\O(N +k^2/\delta)$ & Apdx~\ref{subsec:RelClustResOutliers}\\
\hline
\end{tabular}
\caption{Summary of our main results. $\gamma=o(m)$ is any sublinear approximation factor, $f$ is the maximum number of sets that an element belongs to, $g=O(1)$ is the number of relations in $\allrel$, $d=|\allattr|$, $\delta=z/|\Q(\I)|$, $\beta_1=\min\{n,km\}, \beta_2=\{n,k^z,km\},\beta_3=\min\{m,kz\}, \beta_2'=\min\{n,kz\}$, $\eps$ is any constant in $(0,1)$, $\zeta$ is any positive real number, $\poly(n,m)$ is any polynomial function on $n,m$, and $\O(\cdot)$ skips $\log^{O(1)} n$ or $\log^{O(1)} N$ factors, where $n=|P|$, $N=|\I|$.}
\label{tab:3col8row}
\vspace{-1.5em}
\end{table}

    \begin{itemize}[leftmargin=*]
        \item  For the general disjoint case ($\prob$, $f=1$), we design a $(2,2,O(1))$-approximation algorithm (Section~\ref{subsec:GeneralDS}) that runs in $O\left(\left(nk+k^2m+\beta_1^2+(\beta_2+\beta_3)^{\expon}\right)\cdot \log \beta_1\right)$ time, where $\beta_1=\min\{n,km\}$, $\beta_2=\min\{n,k^2z, km\}$, and $\beta_3=\min\{m,kz\}$. The main idea is to construct a small subset $P'\subseteq P$ and a small subset of $\mathcal{H}'\subseteq \mathcal{H}$ such that solving the $\mathsf{LP}$-based algorithm from Section~\ref{subsec:GeneralIS} on $P', \mathcal{H}'$ is sufficient to get a good approximation for the original $\prob_{k,z}(P,\mathcal{H})$ problem.

        \item For the general $\gprob$ problem ($f>1$), instead of using an $\mathsf{LP}$ solver to solve the linear program, we use the \emph{Multiplicative Weight Update} (MWU) method to approximate its solution. A straightforward implementation of the MWU approach on our problem would lead to super-quadratic running time. Instead, in Section~\ref{subsec:GeometricIS} we use geometric data structures (range trees and BBD trees) to implement the MWU method in near-linear time with respect to $n$ and $m$ . For any small constant $\eps\in (0,1)$, we design a $(2+\eps,2f,2+\eps)$-approximation algorithm for the $\gprob(P,\setrects)$ problem in $O((k+z)(n+m)\log^{d+2}n)$ time.
        
        \item For the disjoint $\gprob$ ($f=1$), we use geometric computing and data structures to accelerate the coreset construction from Section~\ref{subsec:GeneralDS}. For any small constant $\eps\in (0,1)$, we design a $(2+\eps,2,O(1))$-approximation algorithm for the $\gprob(P,\setrects, k, z)$ problem in $O(((n+m)\log^{d-1}n +km\log n+\beta_1\log(\beta_1)+(k+z)(\beta_2+\beta_3)\log^{d+2}(\beta_2))\cdot \log n)$ time , where $\beta_1=\min\{n,km\}$, $\beta_2=\min\{n,kz\}$, and $\beta_3=\min\{m,kz\}$.
    \end{itemize}

    \item \textbf{New Algorithms for Relational $k$-center Clustering.}
    In the relational setting, we design approximation algorithms for the  $\relclustertwo$ with tuple outliers, $\relclusterthree$ with tuple outliers of one specific relation, and $\relclusterone$ problem with result outliers. Specifically:
    \begin{itemize}[leftmargin=*]
       
        \item In Section~\ref{subsubsec:TupleOutliersOneRel}, we focus on the $\relclusterthree(\Q,\I,k,z)$ problem for an acyclic join query $\Q$ and show that the $\relclusterthree$ problem can be reduced to $\gprob$ problem over $\Q(\I)$ and a set of $O(N)$ non-intersecting hyper-rectangles. By designing a relational implementation of our coreset construction for $\gprob$ $f=1$ case, we obtain an $(2+\eps,2,O(1))$-approximation algorithm for the $\relclusterthree(\Q,\I,k,z)$ problem for any small constant $\eps\in(0,1)$, in $O(k^2N^2\log N)$ time.


    \item In Section~\ref{subsubsec:TupleOutliersMultRels} we consider the $\relclustertwo$ problem.
    The $\relclustertwo$ problem can be mapped to the $\gprob$ problem with $f=N$, hence a potential relational implementation of our algorithm for $\gprob$ would lead to a bad approximation on the number of outliers.
    Instead, for any acyclic join query $\Q$, we design a randomized fixed-parameter tractable (FPT) algorithm in $k$ and $z$. In particular, we design an $(1,g,O(1))$-approximation algorithm for the $\relclustertwo(\Q,\I,k,z)$ problem, with probability at least $1-\frac{1}{N}$, with running time $O(2^{g\cdot k+z}\cdot(z+k^d)\cdot N\cdot \log^3 N)$.

     \item Finally, we show that $\relclusterone(\Q,\I,k,z)$ can be mapped to the standard $k$-center clustering with outliers on $\Q(\I)$.
     For an acyclic join query $\Q$, we design a relational implementation of the algorithms in~\cite{charikar2003better, charikar1999constant} and for any small constant $\eps\in(0,1)$, we obtain an $(1,1+\eps,3+\eps)$-approximation algorithm for the $\relclusterone(\Q,\I,k,z)$ problem, with probability at least $1-\frac{1}{N}$, in $O(N\log N +\frac{k^2}{\delta}\log^3 N)$ time, where $\delta=\frac{z}{|\Q(\I)|}$. The details can be found in Appendix~\ref{subsec:RelClustResOutliers}.

    \end{itemize}
Our relational algorithms can be extended to cyclic join queries as briefly discussed in Section~\ref{subsec:cyclic}.
\end{itemize}

\vspace{-1em}
\section{Clustering with set outliers}
We begin our discussion by noticing an interesting connection between the $\prob$ and the classic \emph{set cover} problem. Let $D$ be the sorted list of all pairwise distances between elements in $P$. Observe that the optimal clustering cost $\rho_{k,z}^*(P, \mathcal{H})$ is one of the values in $D$, so we can `guess' this
optimal value by binary searching over all $r \in D$. For a given $r \in D$, suppose we construct a collection $\mathcal{C}$ of $n$ sets, one for each $p_i \in P$ as center and consisting of all $p_j$ such that $\dist(p_i, p_j) \leq r$. Therefore, we now have two collections of sets: $\mathcal{C}$ and $\mathcal{H}$ and the $\prob$ problem is essentially the same as selecting at most $k$ sets from $\mathcal{C}$ and at most $z$ sets from $\mathcal{H}$ that cover all elements in $P$.
One approach to solve $\prob$ is by extending traditional approximation algorithms for set cover in this setting but that will likely result in a $\log n$ approximation on both number of centers and outliers~\cite{dehghankar2024fair}.
Interestingly, as shown next, the set-cover problem can also be reduced to $\prob$ demonstrating that the difficulty of $\prob$ comes from the structure of outlier sets and there is no hope for efficient approximation algorithms in the general case.
We describe our hardness construction below and design efficient (and in some cases near-optimal) approximation algorithms in subsequent sections.




\subsection{Hardness}
\label{sec:hardness}
We show an inapproximability result for the $\prob(P,\mathcal{H}, k,z)$ problem: There is no polynomial time algorithm that returns a valid solution $C, H$ with a sublinear approximation on the cost of the objective function, where $|H|\leq (f-\zeta)z$, for any $\zeta>0$, if the unique games conjecture is true.

Let $\mathsf{SC}(X,\mathcal{Y})$ be an instance of the set-cover problem where $X$ is a set of $n'$ elements and $\mathcal{Y}$ is a family of $m'$ subsets of $X$. Every element in $X$ belongs to at least one set in $\mathcal{Y}$. In low-frequency systems, we assume that every element in $X$ belongs to at most $f=o(\log (n'))$ sets in $\mathcal{Y}$.
A family of sets $Y'\subseteq \mathcal{Y}$ is called a cover if for every $x\in X$, $x\in\bigcup_{y\in Y'}y$. 
An algorithm is a $\lambda$-approximation for the $\mathsf{SC}(X,\mathcal{Y})$ problem if it returns a cover $Y\subseteq \mathcal{Y}$ with $|Y|\leq \lambda\cdot \mathsf{opt}$, where $\mathsf{opt}$ is the cardinality of the minimum cover.
If the unique games conjecture is true there is no $(f-\zeta)$-approximation for the $\mathsf{SC}(X,\mathcal{Y})$ problem, for any $\zeta>0$~\cite{khot2008vertex, dinur2003new}.
In Appendix~\ref{appndx:hardness} we show a polynomial time reduction from the $\mathsf{SC}$ problem to $\prob$ that leads to the proof of the next lemma.

\renewcommand{\H}{\mathcal{H}}

\begin{lemma}
\label{lem:hardness}
   If $\mathsf{Alg}$ is a polynomial time $(1,f-\zeta,\gamma)$-approximation algorithm for the $\prob$ problem, where $\gamma=o(|\mathcal{H}|)$ is any sublinear value, then there exists a polynomial time $(f-\zeta)$-approximation algorithm for the $\mathsf{SC}$ problem.
\end{lemma}
\begin{proof}[Sketch of the proof]
    The full proof is shown in Appendix~\ref{appndx:hardness}. 
    Given $X,\mathcal{Y}$, we create a set $P$ and a family of sets $\H$ as the input for out $\prob$ problem.
    For each $x_i\in X$ we create a point $p_i$ on the real line with coordinate $i$.
    Let $\hat{P}$ be these points on the real line. For an arbitrary value $k=o(|X|)$, let $\hat{Q}=\{q_1,\ldots, q_k\}$ be $k$ points on the real line such that each point $q_j\in \hat{Q}$ is a point on the real line with coordinate $2|X|+j$. The set $P$ is defined as $P=\hat{P}\cup\hat{Q}$. For every set $y\in \mathcal{Y}$, we create a set $\hat{y}=\{p_i\mid x_i\in y\}$ and let $\hat{\mathcal{Y}}=\{\hat{y}\mid y\in \mathcal{Y}\}$. Finally, for every $q_j\in \hat{Q}$ we create a set $\hat{h}_j=\{q_j\}$ and let $\hat{H}=\{\hat{h}_j\mid q_j\in \hat{Q}\}$. We define $\H=\hat{\mathcal{Y}}\cup \hat{H}$.
    We show that if $z\geq \opt$, then $\rho^*_{k,z}(P,\H)=0$.
    Since $\gamma$ is a relative approximation factor, the solution $C_z, H_z$ returned by $\mathsf{Alg}$ should also satisfy $\rho(C_z,P\setminus(\bigcup_{h\in H_z}h))=0$.
    Furthermore, for $z\geq \opt$, without loss of generality, $C_z= \hat{Q}$ and $H_z\subseteq \hat{\mathcal{Y}}$, i.e., $C_z$ contains all points in $\hat{Q}$ and sets in $H_z$ are sets in $\hat{\mathcal{Y}}$ that cover al points in $\hat{P}\subseteq P$.

    For each $z=\{1,\ldots, |\mathcal{Y}|\}$, we run $\mathsf{Alg}$ for $\prob(P,\H,k,z)$ until we find the smallest $z$, say $z'$, such that
    $\rho(C_{z'},P\setminus (\bigcup_{h\in H_z}h))=0$.
    We have $z'\leq \opt$. By definition, $|C_{z'}|\leq k$, and $|H_{z'}|\leq (f-\zeta)\opt$. For every set in $H_{z'}$ there exists a corresponding set in $\mathcal{Y}$ and the union of them cover all points in $X$. Hence, we return a $(f-\zeta)$-approximation for the $\mathsf{SC}$ problem.
\end{proof}

\vspace{-1em}
\begin{theorem}
    If the Unique games conjecture is true, there exists no $(1,f-\zeta,\gamma)$-approximation algorithm that runs in polynomial time for the $\prob$ problem,  where $f = o(\log n)$, $\gamma=o(|\mathcal{H}|)$ is any sublinear value and $\zeta$ any positive real number.
\end{theorem}

\subsection{$\mathsf{LP}$-based Approximation Algorithm for General $\prob$}
\label{subsec:GeneralIS}
We first start with some useful notation. For an element $p_0\in P$ and radius $r$, let $B(p_0,r)=\{p\in P\mid \dist(p,p_0)\leq r\}$.
For an element $p_i\in P$, let $L_i=\{h_j\in \mathcal{H}\mid p_i\in h_j\}$.

\paragraph{Algorithm}
Let $D$ be the sorted list of all pairwise distances. We run a binary search on $D$. Let $r$ be the current distance in the binary search.
We formulate the following integer program (IP).
Let $x_1,\ldots, x_n$ be $n$ variables such that $x_i=1$ if $p_i\in P$ is selected as one of the centers, and $x_i=0$, otherwise.
 Let $y_1,\ldots,y_m$ be $m$ variables such that $y_j=1$ if $h_j\in \mathcal{H}$ is selected as an outlier, and $h_j=0$, otherwise. The third type of constraints, requires that every ball with center a point $p_i$ with radius $r$ must contain a point in the solution of $\prob$ otherwise $p_i$ must belong to at least one outlier set.
 We define the linear program ($\mathsf{LP}$\ref{lp:1}) by replacing the last constraints with $x_i,y_j\in[0,1]$ for every $p_i\in P$ and every $h_j\in \mathcal{H}$.
 

\begin{minipage}[t]{0.48\textwidth}
\begin{align*}
\mathsf{(IP)~~~~} \sum_{i\in[n]}x_i &\leq k \\
\sum_{j\in[m]}y_j &\leq z \\
\sum_{h_j\in L_i}y_j + \sum_{p_\ell\in B(p_i,r)}x_\ell &\geq 1, \quad \forall p_i \in P \\
x_i, y_j \in \{0,1\}, &\quad \forall p_i \in P,\, h_j \in \mathcal{H}
\end{align*}
\end{minipage}
\hfill
\begin{minipage}[t]{0.48\textwidth}
\begin{align}
\lpproblemtext \sum_{i\in[n]}x_i &\leq k \label{eq1} \\
\sum_{j\in[m]}y_j &\leq z \label{eq2} \\
\sum_{h_j\in L_i}y_j + \sum_{p_\ell\in B(p_i,r)}x_\ell \geq 1, &\quad \forall p_i \in P \label{eq3} \\
x_i, y_j \in [0,1], \quad &\forall p_i \in P,\, h_j \in \mathcal{H} \label{eq4}
\end{align}
\vspace{0.5em}
\end{minipage}

If the LP is infeasible, we stop the execution with the current $r$, and continue the binary search with distances larger than $r$. 
If the LP is feasible, we continue as follows.
For every $i\in [n]$, let $x_i^*$ be the value of the variable $x_i$, and for every $j\in[m]$, let $y_j^*$ be the value of the variable $y_j$, in the solution of the LP.
We get the family of sets $H=\{h_j\mid y_j^*\geq \frac{1}{2f}\}$.
Let $\mathsf{Act}=P\setminus (\bigcup_{h\in H}h)$ be the set of active elements after removing the elements that belong to outlier sets in $H$.
Let $C=\emptyset$.
As long as $\mathsf{Act}$ is not empty, we repeat:
Let $p_i$ be an element in $\mathsf{Act}$. We add $p_i$ in $C$, i.e., $C=C\cup\{p_i\}$.
We remove from $\mathsf{Act}$ all points in $B(p_i,2r)$, i.e., $\mathsf{Act}=\mathsf{Act}\setminus B(p_i,2r)$.
We continue the binary search with distances smaller than $r$.
In the end, we return the set of centers $C$ and the family of outliers $H$ derived after the last feasible LP we encountered.

\paragraph{Correctness}
\begin{lemma}
    \label{lem:master}
   It holds that: i) If $r\geq \rho^*_{k,z}(P,\mathcal{H})$ then ($\mathsf{LP}$\ref{lp:1})\ is feasible. ii) $|H|\leq 2fz$. iii)  If ($\mathsf{LP}$\ref{lp:1})\ is feasible for a distance $r$, then  $\rho(C,P\setminus(\bigcup_{h_j\in H}h_j))\leq 2r$. iv) $|C|\leq 2k$.
\end{lemma}
\vspace{-0.5em}
\begin{proof}

We first show i).
Equivalently, we prove that if $r\geq \rho^*_{k,z}(P,\mathcal{H})$ then the corresponding IP is feasible. Since $r\geq \rho^*_{k,z}(P,\mathcal{H})$, there exists the optimum set of $k$ centers $C^*\subset P$ and a family of $z$ sets $H^*\subseteq \mathcal{H}$ such that $\rho(C^*,P\setminus  \bigcup_{h\in H^*}h)=\rho^*_{k,z}(P,\mathcal{H})$.
    Based on this optimum solution, we construct a solution to the IP as follows. If $p_i\in C^*$ then we set $\bar{x}_i=1$, otherwise we set $\bar{x}_i=0$. If $h_j\in H^*$ we set $\bar{y}_j=1$, otherwise $\bar{y}_j=0$. Since $|C^*|=k$ and $|H^*|=z$, the first two constraints are satisfied in the IP. Next, we focus on the third family of constraints in the IP. If $p_i\in \bigcup_{h_j\in H^*}h_j$ then there exists a set $h_{j_i}\in L_i$ such that $h_{j_i}\in H^*$ and $\bar{y}_{j_i}=1$. Hence $\sum_{h_j\in L_i}\bar{y_j}\geq 1$ and the inequality is satisfied. If $p_i\notin \bigcup_{h_j\in H^*}h_j$, then $p_i$ must be covered within distance $r$ from a center in $C^*$. Hence there exists an element $p_{i^*}\in C^*$ such that $\dist(p_i,p_{i^*})\leq r$ where $\bar{x}_{i^*}=1$, so the inequality is satisfied.

Then, we show ii).
Assume that $|H|>2fz$. Then there would be more than $2fz$ variables $y_j$ with value $y_j^*\geq\frac{1}{2f}$, so $\sum_{j\in[m]}y_j^*>z$, which is a contradiction.

Next, we prove iii).
    By definition, the algorithm only finishes when no elements in $\mathsf{Act}$ remain. Every time we add an element $p_i$ in $C$ we remove all points in (current) $\mathsf{Act}$ within distance $2r$ from $p_i$. So all elements in $P\setminus(\bigcup_{h_j\in H}h_j)$ are within distance $2r$ from $C$.

Finally, we prove iv).
    Let $p_i$ be an element that does not belong in any outlier set in $H$. By definition, $\sum_{h_j\in L_i}y_j^*\leq f\frac{1}{2f}=\frac{1}{2}$ so from the third family of inequalities in the LP, we have $\sum_{p_\ell\in B(p_i,r)}x_\ell^*\geq \frac{1}{2}$.
We use an assignment argument. Every time we add a point $p$ in $C$ we assign the weight of all elements within distance $r$ to $p$.
Assume that an element $p_i$ is added in $C$ in one of the iterations of the algorithm.
The weight $w_{p_i}=\sum_{p_\ell\in B(p_i,r)}x_\ell^*$ is assigned to $p_i$. We have that $w_{p_i}\geq \frac{1}{2}$ while $\sum_{i\in[n]}x_i^*\leq k$.
Furthermore, for two different elements $p_i, p_j\in C$, we have that $B(p_i,r)\cap B(p_j,r)=\emptyset$ because $\dist(p_i,p_j)>2r$ since the algorithm removes all active elements within distance $2r$ from a selected element in $C$.
Hence, each time that we add a new element $p_i$ in $C$ a weight $w_{p_i}\geq 1/2$ is assigned to $p_i$. The total amount of weight among all elements is at most $k$, and no fraction of weight is assigned to two different elements in $C$.
Hence $\sum_{p\in C}w_p\leq k$ and $\sum_{p\in C}w_p\geq |C|\cdot \frac{1}{2}$ leading to $|C|\frac{1}{2}\leq k\Leftrightarrow |C|\leq 2k$.
\end{proof}
\vspace{-0.5em}

\paragraph{Running time}
We sort all possible distances in $O(n^2\log n)$ time.
For every element $p_i\in P$, we need $O(m)$ time to identify the family $L_i$.
Then, for each distance $r$ we consider from the binary search, we construct the LP in $O(n^2)$ time computing $B(p_i,r)\cap P$ for every $p_i\in P$. The LP has $O(n+m)$ variables and $O(m+n)$ constraints.
Using~\cite{jiang2021faster}, we solve the LP in roughly $O((m+n)^{\expon})$ time.
Hence the asymptotic running time of the algorithm is $O((n+m)^{\expon}\log n)$.

\begin{theorem}
\label{thm:generalIS}
  Given a set of $n$ elements $P$, a family of $m$ subsets $\mathcal{H}$ over $P$, and positive integer parameters $k, z$, there exists 
  a $(2,2f,2)$-approximation algorithm for the $\prob(P,\mathcal{H},k,z)$ problem that runs in
  $O((n+m)^{\expon}\log n)$ time.
  \end{theorem}

\subsection{Approximation Algorithm for Disjoint $\prob$}
\label{subsec:GeneralDS}
Next, we assume that all pairwise sets in $\mathcal{H}$ are disjoint. In this case notice that $f=1$. Taking advantage of the independent sets, we apply additional optimizations to improve the running time.

For every distance $r$ in the binary search, before we construct the LP, we compute a \emph{coreset} for the $\prob$ problem.
At a high level, we compute a small set $P'\subseteq P$ a small subset $\mathcal{H}'\subseteq{H}$ and parameters $k',z'$ such that a $(O(1),O(1),O(1))$-approximation on $\prob(P',\mathcal{H}',k,z)$ can be used to  derive a $(O(1),O(1),O(1))$-approximation solution on the original $\prob(P,\mathcal{H},k,z)$ problem, efficiently.


\paragraph{Coreset construction}
The construction of our coreset works in two phases. In the first phase we compute a set $\hat{P}\subseteq P$ and a family of sets $\bar{\mathcal{H}}\subseteq \mathcal{H}$ by running a $k$-center algorithm independently to each set $h\in\mathcal{H}$ and keeping only centers that can cover all points in the set within distance $2r$.
In the second phase we continue by removing additional elements where a sufficiently large ball around them intersects elements from sufficiently many different sets. If a set becomes empty in the process is removed from $\mathcal{H}$. In the end of the second phase we get a set $P'\subseteq \bar{P}$ and a family of set $\mathcal{H}'\subseteq \bar{\mathcal{H}}$, which is given as input to the algorithm from Section~\ref{subsec:GeneralIS}.

For every set $h_j\in \mathcal{H}$, we run the Gonzalez algorithm~\cite{gonzalez1985clustering} for $k$-center clustering (without outliers) on the elements in $h_j$. Let $C_j$ be the set of $k$ centers returned by the Gonzalez algorithm.
If $\rho(C_j,h_j)>4r$, then we conclude that $h_j$ cannot be covered with $k$ balls of radius $r$ so $h_j$ must be an outlier. Hence $h_j$ is removed from $\mathcal{H}$ and we set $z\leftarrow z-1$.
Let $H_0$ be the family of removed sets from $\mathcal{H}$.
If $\rho(C_j,h_j)\leq 4r$ then we only keep the elements $C_j$ from the set $h_j$ and we remove all elements in $h_j\setminus C_j$ from $P$, i.e., $P\leftarrow P\setminus (h_j\setminus C_j)$.
Then if two elements in $C_j$ have distance less than $4r$ we keep only one of them in $C_j$. We can do it by visiting each (remaining) element $c\in C_j$ one by one. For every element $c'\in C_j\setminus\{c\}$ if $\dist(c,c')\leq 4r$ then we remove $c'$ from $C_j$.
Let $\bar{P}$ be the remaining set of elements, $\bar{\mathcal{H}}$ be the remaining family of sets and $\bar{z}$ be the updated value of $z$.
If $\bar{z}<0$ then we continue the binary search with larger values of $r$.
If $\beta_1=\min\{n,k\cdot m\}$ we note that $|\bar{P}|=O(\beta_1)$.
Notice that $\rho(C_j,h_j)\leq 8r$ and if 
$r\geq \rho^*_{k,z}(P,\mathcal{H})$ then
$\rho^*_{k,\bar{z}}(\bar{P}, \bar{\mathcal{H}})\leq 9r$.
For a given radius $r$, Gonzalez algorithm runs in $O(n\cdot k)$ time. The procedure removing elements in $C_j$ within distance $4r$ runs in $O(k^2)$ time, so $O(k^2m)$ time in total.

Let $X=\emptyset$ be an empty family of sets.
For every element $p_i\in \bar{P}$ we compute $B_i=B(p_i,18r)\cap \bar{P}$ and $B_i'=B(p_i,36r)\cap \bar{P}$. If $B_i$ contains elements from more than $\bar{z}$ sets, then we add the set $B_i'$ in $X$ and we remove $B_i'$ from $\bar{P}$, i.e., $X=X\cup \{ B_i'\}$ and $\bar{P}=\bar{P}\setminus B_i'$. Furthermore, we update $k\leftarrow k-1$. 
We continue this process until for every $p_\ell\in \bar{P}$, the ball $B_\ell$ contains elements from at most $\bar{z}$ sets, or until $\bar{P}=\emptyset$.
In the end of this process, if a set $h_j\in \bar{\mathcal{H}}$ became empty, then we remove $h_j$ from $\bar{\mathcal{H}}$.
This procedure is executed in $O(\beta_1^2)$ time because we go through $O(\beta_1)$ points in $\bar{P}$ and we compute each set $B_i$ and $B_i'$ in $O(\beta_1)$ time.

Let $P'\subseteq \bar{P}$ be the remaining set of elements, $\mathcal{H}'\subseteq \bar{\mathcal{H}}$ be the new family of remaining set, $z'=\bar{z}$ and $k'$ be the updated number of centers.
If $|\mathcal{H}'|>\min\{m,2kz\}$ then skip the execution and continue the binary search with larger values of $r$.
\vspace{-0.2em}
\begin{lemma}
\label{helper}
If $r\geq \rho^*_{k,z}(P,\mathcal{H})$, then
$|P'|=O(\min\{n,k^2z, km\})$, $|\mathcal{H}'|=O(\min\{m,kz\})$, and $\rho^*_{k',z'}(P',\mathcal{H}')\leq 18r$.
\end{lemma}
\vspace{-1em}
\begin{proof}
    We first show that if 
 $r\geq \rho^*_{k,z}(P,\mathcal{H})$, then 
 $\rho^*_{k',z'}(P',\mathcal{H}')\leq 18r$. 
 We have that $\rho^*_{k,\bar{z}}(\bar{P},\bar{\mathcal{H}})\leq 9r$. Notice that $P'\subseteq \bar{P}$ and $\mathcal{H}'\subseteq \bar{\mathcal{H}}$. Let $\bar{C}\subseteq \bar{P}$ and $\bar{H}\subseteq \bar{\mathcal{H}}$ with $|\bar{C}|\leq k$ and $|\bar{H}|\leq \bar{z}$, such that $\rho(\bar{C},\bar{P}\setminus\bigcup_{h\in \bar{H}})\leq 9r$. For each $h\in\bar{H}$, if $h\notin \mathcal{H}'$ then we remove it from $\bar{H}$. Since $P'\subseteq \bar{P}$ and $\mathcal{H}'\subseteq \bar{\mathcal{H}}$ we have that $\rho(\bar{C},P'\setminus\bigcup_{h\in \bar{H}})\leq 9r$. However $\bar{C}$ is not a valid solution for $\prob_{k',z'}(P',\mathcal{H}')$ because it is not necessary that $\bar{C}\subseteq P'$ and $|\bar{C}|$ might be larger than $k'$. If $X=\emptyset$, then the result follows because $\bar{P}=P'$ and $z=z'$. If $X$ is not empty then it contains some balls of radius $36r$. 
For every $B_\ell'\in X$, notice that the subset $B_\ell\subseteq B_\ell'$ must contain at least one point $q_\ell\in B_\ell$ such that there exists a center $c_\ell\in \bar{C}$ with $\dist(c_\ell,q_\ell)\leq 9r$. This is because the set $B_\ell$ intersects points from more than $\bar{z}$ sets so it is not possible that all of these sets are selected as outliers.
Hence, every $B_\ell'\in X$ fully contains at least one cluster (all points assigned to center $c_\ell$).
Let $\bar{C}\leftarrow\bar{C}\setminus\bigcup_{B_\ell'\in X}c_\ell$. 
Notice that now, $|\bar{C}|=k'$. If there is still a center $c$ from $\bar{C}$ in a ball $B_\ell'\in X$ then we replace $c$ with one point in $P'\setminus \bigcup_{h\in \bar{H}} h$ outside of every ball in $X$ that had $c$ as the closest center in $\bar{C}$ (if such a point exists). By definition, $\bar{C}\subseteq P'$, it contains at most $k'$ points and every point in $P'\setminus \bigcup_{h\in \bar{H}}h$ is within distance $2\cdot 9\cdot r$ from the closest center in $\bar{C}$. Hence, $\rho^*_{k',z'}(P',\mathcal{H}')\leq 18r$.

Let $C'\subseteq P'$ and $H'\subseteq \mathcal{H}'$ are selected such that $|C'|\leq k'$, $|H'|\leq z'$, and $\rho(C',P'\setminus \bigcup_{h\in H'}h)\leq 18r$.
Since $\bigcup_{c'\in C'}(B(c',18r)\cap P')=P'\setminus \bigcup_{h\in H'}h$
and every set $B(c',18r)\cap P'$ contains points from at most $z$ sets, we have that $|\mathcal{H}'|\leq \min\{m,zk'+|H'|\}=O(\min\{m,kz\})$. Every set $h\in \mathcal{H}'$ contains at most $k$ points so $|P'|=O(\min\{n,k^2z,km\})$.
\end{proof}

\vspace{-1em}
\paragraph{Algorithm on coreset}
Next, we show that if we run our algorithm from Section~\ref{subsec:GeneralIS} solving an instance of $\prob(P',\mathcal{H}', k',z')$, we can get a good approximation for the $\prob(P,\mathcal{H}, k,z)$.
We slightly modify the third family of constraints in ($\mathsf{LP}$\ref{lp:1}) to explicitly express the fact that the optimum cost of $\prob(P',\mathcal{H}')$ is at most $18r$. The new LP is shown in ($\mathsf{LP}$\ref{lp:2}).
\vspace{-1em}

\lpproblemlabel{2}
\begin{align*}
\lpproblemtext \sum_{p_i\in P'}x_i &\leq k' \\
\sum_{h_j\in\mathcal{H}'}y_j &\leq z' \\
\sum_{j:\exists h_j \text{ s.t. } p_i\in h_j}y_j + \sum_{p_\ell\in B(p_i,18r)\cap P'}x_\ell &\geq 1, \forall p_i\in P' \\
x_i, y_j \in [0,1], \quad \forall p_i&\in P',\; h_j\in \mathcal{H}'
\end{align*}

Notice that the structure of ($\mathsf{LP}$\ref{lp:2}) is equivalent to ($\mathsf{LP}$\ref{lp:1}). Hence, we run the algorithm from Subsection~\ref{subsec:GeneralIS}.
The only difference is that when we add a point $p_i$ from $\mathsf{Act}$ in $C$ we remove all points within distance $36r$ from $p_i$, instead of $2r$.
Let $\hat{C}$ be these selected set of centers from our algorithm and $\hat{H}$ be the selected outliers. We set $H=\hat{H}\cup H_0$ and $C=\hat{C}\bigcup_{B_{\ell}'}p_\ell'$, where $p_\ell'$ is a point in $B_\ell'$ that does not belong in any set in $H$. If such a point does not exist then $p_\ell'=\emptyset$.
We return the sets $C$ and $H$ we derived after the last feasible LP we encountered. The proof of the next Theorem is shown in Appendix~\ref{appndx:generalDS}.


\vspace{-0.1em}
  \begin{theorem}
  \label{thm:resDS}
  Given a set of $n$ elements $P$, a family of $m$ disjoint subsets $\mathcal{H}$ over $P$, and positive integer parameters $k,z$, there exists a
  $(2,2,O(1))$-approximation algorithm for the $\prob(P,\mathcal{H},k,z)$ problem that runs in $O\left(\left(nk+k^2m+\beta_1^2+(\beta_2+\beta_3)^{\expon}\right)\cdot \log \beta_1\right)$ time, where $\beta_1=\min\{n,km\}$, $\beta_2=\min\{n,k^2z, km\}$, and $\beta_3=\min\{m,kz\}$.
  \end{theorem}
  \vspace{-0.5em}

\vspace{-0.2em}
\section{Geometric clustering with rectangular outliers}
In this subsection we consider the case where $P\subset \Re^d$ and $\setrects$ is a set of hyper-rectangles in $\Re^d$.
We show near linear time algorithms with respect to both $n, m$.
Again, we consider two separate cases, one where $f=1$, and one where $f>1$. We note that in this section each element $\rec\in \setrects$ is a hyper-rectangle, not a subset of $P$, so we write $\rec\cap P$ to denote the subset of $P$ that lies in $\rec$.
Let $B(p,r)=\{x\in \Re^d\mid \dist(p,x)\leq r\}$ the ball with center $p\in \Re^d$ and radius $r\in \Re$. We begin our discussion by revisiting some relevant background on the techniques used in this section.

\vspace{-0.3em}
\subsection{Relevant Background for $\gprob$}
\label{sec:prelim}

\newcommand{\var}{\psi}
\paragraph{Multiplicative Weight Update (MWU) method}
The MWU method~\cite{arora2012multiplicative} is used to solve the following linear feasibility problem.
\vspace{-1em}
\begin{equation}
    \label{feasibility}
    \exists \var\in \mathcal{P}: A\var\geq b,
\end{equation}
where $A\in \Re^{m'\times n'}$, $\var\in \Re^{n'}$, $b\in \Re^{m'}$, $Ax\geq 0$, $b\geq 0$, and $\mathcal{P}$ is a convex set in $\Re^{n'}$.
Arora et al.~\cite{arora2012multiplicative} describe an iterative algorithm using a simple \textsf{ORACLE}.
Let \textsf{ORACLE} be a black-box procedure that solves the following single linear constraint for a probability vector $\probvector\in \Re^{m'}$.
\begin{equation}
\label{eq:oracle}
\exists \var\in \mathcal{P}: \probvector^\top A\var\geq \probvector^\top b.    
\end{equation}

The \textsf{ORACLE} decides if there exists an $\var$ that satisfies the single linear constraint. Otherwise, it returns that there is no feasible solution.
A $\xi$-\textsf{ORACLE} is an \textsf{ORACLE} such that whenever \textsf{ORACLE} manages to find a feasible solution $\hat{\var}$ to problem~\eqref{eq:oracle}, then $A_i\hat{\var}-b_i\in [-1,\xi]$ for each constraint $i\in[m']$, where $A_i$ is the $i$-th row of $A$, and $b_i$ is the $i$-th element of vector $b$.

The algorithm starts by initializing $\probvector$ to a uniform probability vector with value $1/m'$. In each iteration the algorithm solves Equation~\eqref{eq:oracle}. If \eqref{eq:oracle} is infeasible, we return that the original feasibility problem in Equation~\eqref{feasibility} is infeasible. Let $\var^{(t)}$ be the solution of the problem in Equation~\eqref{eq:oracle} in the $t$-th iteration of the algorithm.
Let $\delta_i=\frac{1}{\xi}(A_i\var^{(t)}-b_i)$. We update $\probvector[i]=(1-\delta_i\cdot\eps/4)\probvector[i]$, where $\probvector[i]$ is the $i$-th element of vector $\probvector$. We continue to the next iteration solving \eqref{eq:oracle} for the new probability vector $\probvector$. After $T=O(\xi\log(m')/\eps^2)$ iterations, if every oracle was feasible, they return $\var^*=\frac{1}{T}\sum_{t=1}^T\var^{(t)}$.
Otherwise, if an oracle was infeasible, they argue that the initial problem is infeasible.
Overall, every algorithm using the MWU method to solve a problem in the form of Equation~\eqref{feasibility} should implement two procedures: $\mathsf{Oracle}(\cdot)$ that implements a $\xi$-\textsf{ORACLE} and $\mathsf{Update}(\cdot)$ that updates the probability vector $\probvector$.
In~\cite{arora2012multiplicative} they prove the following theorem.
\begin{theorem}[\cite{arora2012multiplicative}]
\label{thm:mutli-weights}
Given a feasibility problem as defined above, a parameter $\eps$, a $\xi$-\textsf{ORACLE} implemented in procedure $\mathsf{Oracle}(\cdot)$, and an update procedure $\mathsf{Update}(\cdot)$, there is an algorithm which either finds an $\var$ such that $\forall i$, $A_i\var_i\geq b_i-\eps$ or correctly concludes that the system is infeasible. The algorithm makes $O(\xi\log(m')/\eps^2)$ calls to procedures $\mathsf{Oracle}(\cdot)$ and $\mathsf{Update}(\cdot)$.
\end{theorem}

\paragraph{WSPD}
Using a quadtree~\cite{finkel1974quad}, one can compute a \emph{Well-Separated Pair Decomposition} (WSPD)~\cite{callahan1995decomposition, har2005fast} for a set $P$ of $n$ points in $\Re^d$. Specifically, we can construct a list $\mathcal{L} = \{L_1, \ldots, L_z\}$ of $O(\varepsilon^{-d} n)$ distances such that for every pair $p, q \in P$, there exists some $L_j \in \mathcal{L}$ satisfying $(1 - \varepsilon)\,\dist(p, q) \leq L_j \leq (1 + \varepsilon)\,\dist(p, q)$.
The construction time for $\mathcal{L}$ is $O(\varepsilon^{-d} n \log n)$.


\paragraph{BBD tree}
BBD tree~\cite{arya1998optimal, arya2000approximate} is tree-based data structure that is used for counting, reporting and other aggregation queries over balls and other convex query regions. Given a set $P$ of $n$ points in $\Re^d$ the BBD tree $\tree$ over $P$ is constructed in $O(n\log n)$ time and has space $O(n)$. The height of the tree is $O(\log n)$. Every node $u$ of $\tree$ is associated with a region $\square_u$, which is either a hyper-rectangle or a hyper-rectangle with a hole. A hole is defined by another hyper-rectangle. Given a query ball $B(x,r)$ over a point $x\in \Re^d$ and a distance $r$, and a parameter $\eps\in (0,1)$, a query on $\tree$ returns a set of canonical nodes $\canonical(B(x,r))$ such that $B(x,r)\cap P\subseteq \bigcup_{u\in \canonical(B(x,r))}\square_u\cap P\subseteq B(x,(1+\eps)r)\cap P$. Furthermore, for any pair of nodes $u,v\in\canonical(B(x,r))$, we have that $\square_u\cap \square_v=\emptyset$.
The query procedure runs in $O(\eps^{-d+1}+\log n)$ time and $|\canonical(B(x,r))|=O(\eps^{-d+1}+\log n)$.
If for each node $u$ we store $u.c=|P \cap \square_u|$ or $u.P=P\cap \square_u$, notice that the BBD tree can be used to count or report the points in a ball $B(x,r)$. According to its definition we report (or count) all points in $B(x,r)\cap P$, but we might also report (or count) points of $P$ slightly outside the ball $B(x,r)$.

\paragraph{Range tree}
While BBD tree is used for approximate query processing over a ball, a range tree~\cite{bentley1979decomposable, de2000computational} is used for exact reporting or counting queries over a hyper-rectangle. Given a set $P$ of $n$ points in $\Re^d$, a range tree $\rangetree$ is constructed over $P$ in $O(n\log^{d-1}n)$ time with $O(n\log^{d-1}n)$ space.
Each node $u$ of $\rangetree$ is associated with a hyper-rectangle $\square_u$.
Given a query hyper-rectangle $\rec$ in $\Re^d$, the query procedure on $\rangetree$ returns a set of canonical nodes $\canonical(\rec)$ such that $\rec\cap P=\bigcup_{u\in \canonical(\rec)}(\square_u\cap P)$.
For any pair of nodes $u,v\in\canonical(\square)$, we have that $\square_u\cap \square_v=\emptyset$.
The query procedure runs in $O(\log^d n)$ time and $|\canonical(\square)|=O(\log^d n)$.
If for each node $u$ we store $u.c=|P \cap \square_u|$ or $u.P=P\cap \square_u$, notice that the range tree can be used to count or report the points in $P$ that lie in a hyper-rectangle $\rec$.

\vspace{-0.7em}
\subsection{Approximation Algorithm for General $\gprob$}
\label{subsec:GeometricIS}
\paragraph{High-level ideas}
We design a new algorithm that uses the MWU approach to solve ($\mathsf{LP}$\ref{lp:1}).
While the MWU approach can work directly on ($\mathsf{LP}$\ref{lp:1}), it takes $\Omega((n+m)^2)$ time to run. Instead, we define a new linear feasibility problem, called ($\mathsf{LP}$\ref{lp:3}) and we use the MWU method to approximately solve it in near-linear time. Finally, we round its fractional solution to return a valid solution for the $\gprob(P,\setrects)$ problem in near-linear time using advanced geometric data structures.

\vspace{-1em}
\lpproblemlabel{3}
\begin{align}
\lpproblemtext \sum_{i\in[n]}x_i &\leq k \label{neq1} \\
\sum_{j\in[m]}y_j &\leq z \label{neq2} \\
\sum_{\rec_j\in L_i}y_j + \sum_{p_\ell\in S_{p_i}^\eps}x_\ell &\geq 1, \quad \forall p_i \in P \label{neq3} \\
x_i, y_j \in [0,1], \quad &\forall p_i \in P,\; \rec_j \in \setrects \label{nsimple}
\end{align}

Assume that $r$ is a pairwise distance among the items in $P$.
Instead of solving ($\mathsf{LP}$\ref{lp:1}), we define a slightly different linear feasibility problem 
($\mathsf{LP}$\ref{lp:3}). The Constraints~\eqref{eq1} and Constraints~\eqref{eq2} from ($\mathsf{LP}$\ref{lp:1}) remain the same.
However, we slightly modify Constraints~\eqref{eq3}.
For a point $p_i$ we define a set $S_{p_i}^\eps\subseteq P$ denoting its ``neighboring'' points, with a definition of neighboring which is convenient for the data structure we use. The set $S_{p_i}^\eps$ contains all points within distance $r$ from $p_i$, might contain some points within distance at most $(1+\eps)r$, and no point with distance more than $(1+\eps)r$.
The properties of the BBD tree are used to formally define $S_{p_i}^{\eps}$.
We define $S_{p_i}^{\eps}=\{p\in \square_{u}\cap P\mid u\in \canonical(B(p_i,r))\}$, i.e., the set of points in the canonical nodes returned by the query ball $B(p_i,r)$.
We replace Constraints~\eqref{eq3} with
$\sum_{\rec_j\in L_i}y_{j}+\sum_{p_\ell\in S_{p_i}^\eps}x_\ell\geq 1$.
Overall, the new feasibility problem is given in ($\mathsf{LP}$\ref{lp:3})




Recall that the MWU method solves feasibility problems in the form of Equation~\eqref{feasibility}, $\exists \var\in \mathcal{P}:Ax\geq b$. Next, we show that ($\mathsf{LP}$\ref{lp:3}) can be written in this form by defining $\mathcal{P}$, $A$, and $b$.
Instead of considering that the trivial constraints $\mathcal{P}$ contain only the inequalities~\eqref{nsimple}, we assume that Constraints~\eqref{neq1}  and~\eqref{neq2} are also and contained in $\mathcal{P}$. This will allow us later to design a $(k+z)$-\textsf{ORACLE}. The set $\mathcal{P}$ is convex because it is defined as the intersection of $O(m+n)$ halfspaces $\Re^{m+n}$. Hence, it is valid to use the MWU method.
The new Constraints~\eqref{neq3} define the binary matrix $A$, having one row for every point $p_i\in P$ and one column for every point in $P$ and every hyper-rectangle in $\setrects$. The value $A[i,\ell]=1$ if $p_\ell\in S_{p_i}^\eps$, otherwise it is $0$.
The value $A[i,n+j]=1$ if $\rec_j\in L_i$, otherwise it is $0$.
Notice that each variable $x_\ell$ corresponds to $\var_\ell$, and each variable $y_j$ to $\var_{n+j}$.
Finally, $b$ is defined as a vector in $\Re^d$ with all elements being $1$. From Theorem~\ref{thm:mutli-weights}, we know that the MWU method returns a solution that satisfies the constraints in $\mathcal{P}$ (Constraints~\eqref{neq1} and~\eqref{neq2}) exactly, while the Constraints~\eqref{neq3} are satisfied with an $\eps$ additive error.

A straightforward implementation of the MWU method over ($\mathsf{LP}$\ref{lp:3}) would take super-quadratic time; even the computation of $A$ takes $O(n^2+n\cdot m)$ time. Our new algorithm does not construct $A$ explicitly. Likewise, it does not construct the sets $S_p^\eps$ explicitly. Instead, we use geometric data structures to implicitly represent $A$, $S_p^\eps$ and execute our algorithm in near-linear time.

Before presenting our method, we note that several works in computational geometry have applied the MWU method to implicitly solve special classes of LPs~\cite{chan2020faster, chekuri2020fast, clarkson2005improved}. 
For example,~\cite{chekuri2020fast} presents MWU-based algorithms with runtime linear in the number of non-zeros in matrix $A$. In our case, $A$ may contain $O(n(n+m))$ non-zero entries, ruling out near-linear time guarantees. Moreover,~\cite{chan2020faster, chekuri2020fast} give constant-factor approximations for problems such as geometric set cover or maximum independent set of disks, but there are key differences: (i) their algorithms apply only in $\Re^2$ or $\Re^3$;
ii) the optimization problem in $\prob$ is fundamentally different. Lastly, this section follows the structure of~\cite{kurkure2024faster}, which uses geometric data structures to speed up MWU. However, while they only manage one type of variable ($x$) and use a BBD tree, our setting involves two variable types ($x, y$), requiring a combination of a range tree and a BBD tree.

\paragraph{Algorithm}
We show the pseudocode in Algorithm~\ref{alg:alg0}. First, it computes a sorted array $\Gamma$ of (candidates of) pairwise distances in $P$. Then it runs a binary search on $\Gamma$. Lines 2, 4, 5, 13, 14, 16 are all trivial executing the binary search on $\Gamma$. Each time we find a feasible (infeasible) solution for the optimization problem in the MWU method we try smaller (larger) values of $r$. For a particular 
\begin{wrapfigure}{r}{0.45\textwidth}  
\vspace{-1em}  
\begin{minipage}{0.45\textwidth}
\begin{algorithm}[H]
    \caption{$MWU(P,\setrects,\eps, k, z)$}
    \label{alg:alg0}
    \small
    $\Gamma\!\!\gets\!\!$ Sorted array of pairwise distances in\hspace{0.04em}$P$\;
    $M_l\gets 0$,\hspace{0.25em} $M_u\gets |\Gamma|-1$\;
    $T\gets O(\eps^{-2}\xi\log n)$\;
    \While{$M_l\neq M_u$}{
        $M\gets\lceil (M_l+M_u)/2 \rceil$, \hspace{0.25em}
        $r\gets \Gamma[M]$\;
        $\hat{\var}\gets(0,\ldots,0)^\top\in \Re^{n+m}$\;
        $\probvector\gets (\frac{1}{n}, \ldots, \frac{1}{n})^\top\in \Re^{n}$\;
        \For{$1,\ldots, T$}{
            $\bar{\var}\gets$\textsf{Oracle}$(P,\setrects,r,\eps, k,z)$\;
            \If{$\bar{\var}\neq\emptyset$}{
                $\hat{\var}\gets \hat{\var}+\bar{\var}$\;
                $\probvector\gets$\textsf{Update}$(P, \setrects, \bar{\var},r,\eps)$\;
            }\Else{
                $M_u\gets M-1$,
                Go to Line 4\;
            }
        }
        $\hat{\var}\gets \hat{\var}/T$\;
        $M_l\gets M$\;
    }
    $(C,\rects)\gets$\textsf{Round}$(P, \setrects,\eps, \hat{\var}, k, z)$\;
    \Return $(C,\rects)$\;
\end{algorithm}
\end{minipage}
\vspace{-1.5em}  
\end{wrapfigure}
$r\in \Gamma$, in lines 4--16 we use the MWU method to solve the LP. The algorithm follows the MWU method as described in Section~\ref{sec:prelim}.
In particular, for at most $T=O(\eps^{-2}\xi\log n)$ iterations, it calls
$\mathsf{Oracle}(\cdot)$ in line 9 to decide whether there exists $\var$ such that $\probvector^\top A\var\leq \probvector^\top b$ and $\var\in \mathcal{P}$. In our case we notice that $b\in\Re^{n}$ and $b=\{1,\ldots, 1\}$, so it is sufficient to decide whether there exists $\var$ such that $\probvector^\top A\var\geq \sum_{\ell}\probvector[\ell]\Leftrightarrow \probvector^\top A\var\geq 1$ and $\var\in \mathcal{P}$. If $\mathsf{Oracle}()$ returns a feasible solution $\bar{\var}$, in line 11 it updates the final solution $\hat{\var}$ and uses the $\mathsf{Update}(\cdot)$ procedure to update the vector $\probvector$ based on $\bar{\var}$. If all $T$ iterations return feasible solutions, in line 15 it computes the solution of LP for a given $r$.
In the end, in line 17 we
call a rounding procedure to compute a set of centers $C$ and a set of outliers $\rects$, using the (approximate) solution $\bar{\var}$ from the LP.

If $Q_{\Gamma}, Q_O, Q_U, Q_R$ is the running time to compute $\Gamma$, run $\mathsf{Oracle}(\cdot)$ $\mathsf{Update}(\cdot)$, and $\mathsf{Round}(\cdot)$, respectively, then the overall running time of Algorithm~\ref{alg:alg0} is
$O(Q_{\Gamma} + (\eps^{-2}\xi \log n)(Q_O+Q_U+n)\log |\Gamma| + Q_R)$.
Using a straightforward implementation of these procedures we have $Q_{\Gamma}, Q_O, Q_U, Q_R = \Omega(n^2)$.
Next, we use geometric tools to show how we can find a set $\Gamma$ and run all the procedures in near-linear time with respect to $n$ and $m$. We also show that $\xi=k+z$.

\paragraph{The $\mathsf{Oracle}(\cdot)$ procedure}
We design a $(k+z)$-\textsf{ORACLE} procedure as defined in Section~\ref{sec:prelim}. The goal is to decide whether there exists $\var$ such that $\probvector^\top A\var\leq 1$ and $\var\in \mathcal{P}$. We note that $\mathsf{Oracle}(\cdot)$ does not compute the matrix $A$ explicitly.
In fact, $\probvector^\top A\var$ can be written as $\sum_{p_\ell\in P}\lambda_\ell \var_\ell + \sum_{\rec_j\in \setrects}\mu_{j}\var_{n+j}$, for some real coefficients $\lambda_\ell, \mu_j$.
Intuitively, our goal is to find $k$ points from $P$ and $z$ hyper-rectangles from $\setrects$ having the largest coefficients $\lambda_\ell, \mu_j$. Our algorithm first finds all coefficients $\lambda_\ell, \mu_j$ and then it chooses the $k$ largest $\lambda_\ell$ coefficients from $P$ and the $z$ largest $\mu_j$ coefficients from $\setrects$. In that way, we find the solution $\bar{\var}$ that maximizes $\probvector^\top A\bar{\var}$ for $\bar{\var}\in\mathcal{P}$. Finally, we check whether $\probvector^\top A\bar{\var}\geq 1$.

We are given a probability vector $\probvector\in \Re^n$. Each value $\probvector[i]\in \probvector$ corresponds to the \emph{weight} of the $i$-th row of matrix $A$. In other words each point $p_i\in P$ is associated with a weight $\probvector[i]$.
Let $\tree$ be the BBD tree constructed as described in Section~\ref{sec:prelim} over the set of points $P$.
For every node $u$ of $\tree$, we initialize a weight $u_s=0$.
Let $\rangetree$ be the tree constructed as described in Section~\ref{sec:prelim} over the set of weighted points $P$.
By construction, every node $v$ 
of $\rangetree$, initializes a weight $v_s$ which is equal to the sum of weights of points that lie in leaf nodes of the subtree rooted at $v$.
The data structure $\tree$ has $O(n)$ space and can be constructed in $O(n\log n)$ time. The data structure $\rangetree$ has $O(n\log^{d-1}n)$ space and can be constructed in $O(n\log^{d-1} n)$ time.


\newcommand{\canonicalrange}{\mathcal{U}}

Using $\tree$ and $\rangetree$, 
we show how to check whether there exists $\bar{\var}$ such that $\probvector^\top A\bar{\var}\geq 1$ and $\bar{\var}\in \mathcal{P}$.
For each $p_\ell\in P$ we run the query $\tree(p_\ell, r)$ and we get the set of canonical nodes $\canonical(B(p_\ell,r))$. For each node $u\in\canonical(B(p_\ell,r))$, we update $u_s\leftarrow u_s+\probvector[\ell]$.
After we traverse all points in $P$, 
we revisit each point $p_\ell\in P$ and continue as follows: We initialize a weight $w_\ell=0$. We start from the leaf node of $\tree$ that contains $p_\ell$ and we traverse $\tree$ bottom up until we reach the root of the BBD tree. Let $v$ be a node we traverse; we update $w_\ell=w_\ell+v_s$.
Intuitively in this process we get the sum of weight of all points $p_i$ such that $p_\ell \in S_{p_i}^\eps$.
After computing all values $w_\ell$ for $\ell\in[n]$,
we find the $k$ points from $P$ with the largest weights $w_\ell$. Let $P_L$ be these points. For each $p_\ell\in P_L$ we set $\bar{\var}_\ell=1$. Otherwise, if $p_\ell\notin P_L$ we set $\bar{\var}_\ell=0$.
Next, we process the hyper-rectangles.
For each $\rec_j\in \setrects$, we run a query on $\rangetree$ and we get the set of canonical nodes $\canonicalrange(\rec_j)$.
We set $\tau_j=\sum_{u\in\canonicalrange(\rec_j)}u.s$.
Let $\setrects_L\subseteq \setrects$ be the $z$ hyper-rectangles with the largest weights $\tau_j$.
For each $\rec_j\in \setrects_L$ we set $\bar{\var}_{n+j}=1$. Otherwise, if $\rec_j\notin \setrects_L$ we set $\bar{\var}_{n+j}=0$.
If $\sum_{p_\ell\in P}w_\ell\bar{\var}_\ell +\sum_{\rec_j\in \setrects}\tau_j\bar{\var}_{n+j}\geq 1$ the oracle returns $\bar{\var}$ as a feasible solution. Otherwise, it returns that there is no feasible solution.

\vspace{-0.5em}
\paragraph{Correctness}
We show that the \textsf{ORACLE} we design is correct.
We first show that $\lambda_\ell=w_\ell$ for $\ell\in[n]$ and $\mu_j=\tau_j$ for $j\in[n]$,
so $\probvector^\top A\var = \sum_{p_\ell\in P}w_\ell \var_\ell + \sum_{\rec_j\in \setrects}\tau_j \var_{n+j}$.
Recall that $\probvector^\top A\var=\sum_{p_\ell\in P}\lambda_\ell \var_\ell + \sum_{\rec_j\in \setrects}\mu_j \var_{n+j}$, for the real coefficients $\lambda_\ell$ and $\tau_j$.

By definition, each $\lambda_\ell$ is defined as $\lambda_\ell=\sum_{p_i\in P}\probvector[i]\cdot \mathcal{I}(p_\ell\in S_{p_i}^\eps)$, where $\mathcal{I}(p_\ell\in S_{p_i}^\eps)=1$ if $p_\ell\in S_{p_i}^\eps$ and $0$ otherwise.
If $p_\ell\in S_{p_i}^\eps$ then by definition $\lambda_\ell$ contains a term $\probvector[i]$ in the sum. 
There exists also a node $u\in \canonical(B(p_i,r))$
such that $p_\ell\in \square_u\cap P$, i.e.,  $p_\ell$ lies in a leaf node of the subtree rooted at $u$.
Starting from the leaf containing $p_\ell$, our algorithm will always visit the node $u$, and since $u\in \canonical(B(p_i,r))$ we have that $\probvector[i]$ is a term in the sum $u_s$ so by updating $w_\ell=w_\ell+u_s$ we include the term $\probvector[\ell]$ in the weight $w_\ell$.
Overall, we have that $w_\ell=\lambda_\ell$ for $\ell\in[n]$.
Similarly, each $\mu_j$ is defined as $\mu_j=\sum_{p_i\in P\cap \rec_j}\probvector[i]$. By definition 
$\bigcup_{v\in\canonicalrange(\rec_j)}(\square_v\cap P) = \rec_j\cap P$ and
$\sum_{v\in\canonicalrange(\rec_j)}v.s=\sum_{p_i\in P\cap \rec_j}\probvector[i]$ so $\tau_j=\mu_j$ for $j\in[n]$.
We conclude that our algorithm finds all the correct coefficients in the linear function $\probvector^\top Ax$.

Then we focus on maximizing the sum $\sum_{p_\ell\in P}w_\ell \var_\ell +\sum_{\rec_j\in\setrects}\tau_j \var_{n+j}$ satisfying $\var\in \mathcal{P}$. The constraints in $\mathcal{P}$ suggests that $\sum_{\ell\in[n]}\var_\ell=k$ and $\sum_{j\in[m]}\var_{n+j}=z$. Our algorithm sets the variables $\var_\ell$ that are multiplied by the largest $k$ coefficients to $1$ and sets the variables $\var_{n+j}$ that are multiplied by the largest $z$ coefficients to $1$. Hence, by definition
$\sum_{p_\ell\in P}w_\ell\bar{\var}_\ell +\sum_{\rec_j\in \setrects}\tau_j\bar{\var}_{n+j}$ maximizes the sum $\sum_{p_\ell\in P}w_\ell \var_\ell +\sum_{\rec_j\in\setrects}\tau_j \var_{n+j}$ for $x\in\mathcal{P}$.

Let $\bar{\var}$ be the feasible solution returned by $\mathsf{Oracle}(\cdot)$.
Notice that by definition, $\bar{\var}$ sets $k+z$ variables to $1$. Hence, for each Constraint~\eqref{neq3}, it holds that $A_i\bar{\var}-b_i\leq k-1$ and $A_i\bar{\var}-b_i\geq -1$, where $A_i$ is the $i$-th row of $A$. We conclude that our $\mathsf{Oracle}$ procedure computes a $(k+z)$-\textsf{ORACLE} as defined in~\cite{arora2012multiplicative}, so $\xi=k+z$.

\vspace{-0.5em}
\paragraph{Running time.}
For each new probability vector $\probvector$ we construct $\tree$ in $O(n\log n)$ time and $\rangetree$ in $O(n\log^{d-1}n)$ time. For each point $p_i$ we compute $\canonical(B(p_i,r))$ in $O(\log n + \eps^{-d+1})$ time. Furthermore, the height of $\tree$ is $O(\log n)$ so for each point $p_i$ we need additional $O(\log n)$ time to compute $w_i$.
For each hyper-rectangle $\rec_j$ we compute $\canonicalrange(\rec_j)$ in $O(\log^d n)$ time, so $\tau_j$ is computed in $O(\log^d n)$ time.
After computing the weights, we find the largest $k$ weights $w_i$ and the largest $z$ weights $\tau_j$ in $O(n+m)$ time. Overall, $Q_O=O\left(n\log^{d-1} n +m\log^d n + n\log n + n\eps^{-d+1}\right)=O\left((n+m)\log^d n + n\eps^{-d+1}\right)$.


We show the procedures $\mathsf{Update}(\cdot)$, $\mathsf{Round}(\cdot)$, and the computation of set $\Gamma$ in Appendix~\ref{appndx:geometricIS}.
The running time of these procedures are all bounded by $Q_O$. 
Here we only discuss the high level ideas of these procedures: The $\mathsf{Update}(\cdot)$ operation uses a BBD tree and a range tree similar to the $\mathsf{Oracle}(\cdot)$ procedure to update the vector $\probvector$.
The $\mathsf{Round}(\cdot)$ procedure implements the rounding algorithm in Section~\ref{subsec:GeneralIS}.
It first uses a BBD tree to remove active points within each ball $B(p_i,2r)$.  Furthermore, from the MWU method, Constraints~\eqref{neq3} are satisfied approximately by an additive $\eps$ term. 
This leads to $|C|\leq (2+\eps)k$, similar to the proof of Lemma~\ref{lem:master}.
Finally, the set $\Gamma$ is computed efficiently using a WSPD (Section~\ref{sec:prelim}).


\begin{theorem}
\label{thm:geomfinal}
  Given a set $P$ of $n$ points in $\Re^d$, a set $\setrects$ of $m$ hyper-rectangles in $\Re^d$, where $d=O(1)$, integer parameters $k, z>0$, and a small constant $\eps\in(0,1)$, there exists a
  $(2+\eps,2f,2+\eps)$-approximation algorithm for the $\gprob(P,\setrects,k,z)$ problem that runs in $O\left(\left(k+z\right)\cdot\left(n+m\right)\cdot\log^{d+2} n\right)$ time.
  \end{theorem}
\vspace{-1em}

\subsection{Approximation Algorithm for Disjoint $\gprob$}
\label{subsec:GeometricDS}
Obviously, the algorithm from Theorem~\ref{thm:geomfinal} works when every point belongs to exactly one hyper-rectangle in $\setrects$. In this case ($f=1$),
we get a $(2+\eps,2,2+\eps)$-approximation algorithm for the 
$\gprob(P,\setrects,k,z)$ problem that runs in $O\left(\left(k+z\right)\cdot\left(n+m\right)\cdot\log^{d+2} n\right)$ time.
While this algorithm is already efficient and near linear on both $n$ and $m$, the improvements in Section~\ref{subsec:GeneralDS} can lead to an even faster algorithm for the $\gprob$ problem for $f=1$.
The high level idea is that we use geometric data structures to construct the coreset from Section~\ref{subsec:GeneralDS}, efficiently. Then, we form the linear problem using only the points in the coreset and we use the MWU method from Section~\ref{subsec:GeometricIS} to get an efficient approximation algorithm. We show the details in Appendix~\ref{appndx:geometricDS}.

\begin{theorem}
\label{thm:coresetgeomfinaldisjoin}
  Given a set $P$ of $n$ points in $\Re^d$, a set $\setrects$ of $m$ hyper-rectangles in $\Re^d$, where $d$ is a constant, such that every point in $P$ belongs to one hyper-rectangle from $\setrects$, positive integer parameters $k, z$, and a small constant $\eps\in(0,1)$, there exists a
  $(2+\eps, 2,O(1))$-approximation algorithm for the $\gprob(P,\setrects,k,z)$ problem that runs in $O(((n+m)\log^{d-1}n +km\log n+\beta_1\log(\beta_1)+(k+z)(\beta_2+\beta_3)\log^{d+2}(\beta_2))\cdot\log n)$ time, where $\beta_1=\min\{n,km\}$, $\beta_2=\min\{n,kz\}$, and $\beta_3=\min\{m,kz\}$.
  \end{theorem}

\newcommand{\loutlier}{\sf $k$-Center with Line Outliers}
\section{Relational $k$-center with outliers}
\label{sec:relkcenter}
In this section we study the relational $k$-center problem with outliers.
Our goal is to design fast algorithms without explicitly computing $\Q(\I)$. Multiple combinatorial problems have been studied in this setting such as $k$-center/median/means clustering, diversity, classifiers and more~\cite{agarwal2024computing, esmailpour2024improved, chen2022coresets, curtin2020rk, moseley2021relational, merkl2025diversity,rendle2013scaling, khamis2018ac, kumar2015learning, schleich2016learning,
abo2021relational, abo2018database, yang2020towards, cheng2019nonlinear, cheng2021efficient, schleich2019learning, deep2021ranked, deep2022ranked, tziavelis2020optimal,tziavelis2023efficient, carmeli2023tractable, merkl2025diversity, arenas2024towards, kara2024f}.


Interestingly, without explicitly computing $\Q(\I)$, we can get some important statistics or tuples from the join result via \emph{oracles}.
The next lemmas are shown in~\cite{esmailpour2024improved} and~\cite{agarwal2024computing}.

\begin{lemma}[\cite{esmailpour2024improved}]
\label{lem:Rects}
    Let $\rec$ be a hyper-rectangle in $\Re^d$. For an acyclic join query $\Q$, there exists an oracle $\mathsf{CountRect}(\Q, \I,\rec)$ to count $|\Q(\I)\cap \rec|$ in $O(N\log N)$ time. Furthermore, there exists an oracle $\mathsf{SampleRect}(\Q, \I,\rec, z)$ to sample $z$ samples from $\Q(\I)\cap \rec$ in $O((N+z)\log N)$  time.
\end{lemma}

\begin{lemma}[\cite{agarwal2024computing}]
    \label{lem:relkCenter}
    For an acyclic join query $\Q$ and a small constant $\eps\in(0,1)$,
    there exists an oracle $\mathsf{RelCluster}(\Q,\I,k)$ that returns a set $S\subseteq \Q(\I)$ and a number $r_S$, with $|S|\leq k$, such that $\rho(S,\Q(\I))\leq r_S\leq (2+\eps)\rho^*_k(\Q(\I))$ in $O(k^2\cdot N)$ time.
\end{lemma}


In this section we focus on $\relclusterthree$ and $\relclustertwo$ problems.
We first present all results assuming that $\Q$ is an acyclic join query. In the end, we extend all results for general join queries.


\subsection{Relational $k$-center clustering with tuple-outliers -- Acyclic queries}
\label{subsec:RelClustRupleOutliers}
We first show an interesting  connection of the $\relclustertwo$ problem with the $\gprob$ problem.
The set of points we should cover is $\Q(\I)$.
Let $t\in R_j$ be a tuple in relation $R_j\in \allrel$. Let $\Q_t(\I)=\pi_{\allattr_j=t}(\Q(\I))$, be the set of join results that are generated using the tuple $t$. Every time that $t$ is selected as a tuple outlier, then all points from $\Q_t(\I)$ should be removed from $\Q(\I)$. 
The main observation is that for every $t\in \I$ there exists a (degenerate) hyper-rectangle $\rec_t$ in $\Re^d$ such that $\Q_t(\I)=\rec_t\cap \Q(\I)$. 
For every $A_i\in \allattr$, if $A_i\notin \allattr_j$ then we define the (open) interval $I_i=(-\infty,\infty)$, while if $A_i\in \allattr_j$ then we define the (point) interval $I_i=[\pi_{A_i}(t),\pi_{A_i}(t)]$. We define the degenerate hyper-rectangle $\rec_t=I_1\times I_2\times\ldots\times I_d$. It is easy to observe that $\Q_t(\I)=\rec_t\cap \Q(\I)$.
Thus, the $\relclustertwo$ problem can be mapped to the $\gprob$ problem, where $P=\Q(\I)$ and $\setrects=\bigcup_{t\in \I}\rec_t$. Every time that a hyper-rectangle $\rec_t$ is selected as an outlier in $\gprob$, then tuple $t$ is set to be an outlier in the $\relclustertwo$ problem.

It might be the case that for two tuples $t_1\in R_{i_1}$ and $t_2\in R_{i_2}$, $\rec_{t_1}\cap \rec_{t_2}\cap \Q(\I)\neq \emptyset$, i.e., the sets defined by the hyper-rectangles in $\setrects$ intersect.
Hence, someone might try to use our algorithm from Theorem~\ref{thm:geomfinal} to solve $\relclustertwo$. 
Unfortunately, in addition to the challenges implementing the algorithm from Theorem~\ref{thm:geomfinal} in the relational setting, notice that $f$ might be $O(N)$ leading to a bad approximation factor on the number of outliers.


\subsubsection{\textbf{Tuple outliers from one relation}}
\label{subsubsec:TupleOutliersOneRel}

We assume that only tuples from relation $R_1\in \allrel$ can be marked as outliers. We map $\relclusterthree$ to $\gprob$ where hyper-rectangles in $\setrects$ do not intersect. Indeed consider two different tuples $t, t'\in R_1$. By definition, $\rec_{t}=\bigtimes_{A_i\notin \allattr_1}(-\infty,\infty) \bigtimes_{A_i\in \allattr_1}[\pi_{A_i}(t), \pi_{A_i}(t)]$ and $\rec_{t'}=\bigtimes_{A_i\notin \allattr_1}(-\infty,\infty) \bigtimes_{A_i\in \allattr_1}[\pi_{A_i}(t'), \pi_{A_i}(t')]$. Since $t\neq t'$, there exists at least one attribute $A_{\ell}\in \allattr_1$ such that $\pi_{A_\ell}(t)\neq \pi_{A_\ell}(t')$, so $\rec_t\cap \rec_{t'}=\emptyset$. 
Fortunately, our algorithm for the $\gprob$ runs on a small coreset, when hyper-rectangles in $\setrects$ do not intersect. We construct the coreset efficiently in the relational setting and then we follow the algorithm from Section~\ref{subsec:GeometricDS}.

\paragraph{Algorithm}
Recall that the algorithm from Section~\ref{subsec:GeometricDS} runs a binary search over an approximate set of all pairwise distances (WSPD). Unfortunately, we are not aware of how to construct a WSPD in the relational setting. Instead, we use the trick from~\cite{agarwal2024computing} to run a binary search on the $L_\infty$ pairwise distances of the points in $\Q(\I)$. In fact Agarwal et al.~\cite{agarwal2024computing} showed that the $\ell$-th smallest $L_\infty$ pairwise distance among all points in $\Q(\I)$ can be computed in $O(N\log N)$ time.

Let $r$ be a guess of the cost of $\gprob(\Q(\I),\setrects)$ from the binary search. Recall that the algorithm in Section~\ref{subsec:GeometricDS} runs the standard $k$-center algorithm on the elements of every input set.  In the relational setting, for every $t\in R_1$, we should run the $k$-center clustering on $\rec_t\cap \Q(\I)$. 
We define $R_1'=\{t\}$ and the query $\Q'_t=R_1'\Join R_2\Join \ldots R_g$. Then notice that $\Q'_t(\I)=\rec_t\cap \Q(\I)$. Hence, we run the oracle for the relational $k$-center clustering from Lemma~\ref{lem:relkCenter} on $\Q'_t(\I)$.
Let $C_t$ be the set of $k$ centers and $r_{C_t}$ be the value returned by the oracle in Lemma~\ref{lem:relkCenter}. If $r_{C_t}>2(2+\eps) r$ then we set $C_t=\emptyset$ and we mark $t$ as a tuple-outlier, $T=T\cup\{t\}$. We also remove $\rec_t$ from $\setrects$.
If $r_{C_t}\leq 2(2+\eps)r$ then we continue with the same procedure for the next tuple $t$ in $R_1$. In the end, let $C=\bigcup_{t\in R_1}C_t$, so $|C|=O(kN)$ and $|\setrects|=O(N)$. The algorithm from now on proceeds exactly as in the standard computational setting in Section~\ref{subsec:GeometricDS}.
In the end of the algorithm let $S$ be the set centers  and $T$ be the set of tuple outliers in $R_1$ returned by the algorithm.

The correctness and the running time follows straightforwardly from Theorem~\ref{thm:geomfinal} and Lemma~\ref{lem:relkCenter}.
For completeness we show the details in Appendix~\ref{appndx:tupleoutliers}.

\begin{theorem}
        For a database schema $\allrel$ with $g$ tables $R_1,\ldots, R_g$, an acyclic join query $\Q$ of $d$ attributes, a database $\I$ with $|\I|=N$, integer positive parameters $k, z$, and a small constant $\eps\in (0,1)$,
        there exists a
        $(2+\eps,2,O(1))$-approximation algorithm for the $\relclusterthree(\Q,\I,k,z)$ problem,
        with probability at least $1-\frac{1}{N}$, that runs in $O(k^2N^2\log N)$ time.
\end{theorem}

\subsubsection{\textbf{Tuple outliers from multiple relations}}
\label{subsubsec:TupleOutliersMultRels}
Lastly, we focus on the more involved case assuming that outliers can be selected from multiple (possibly all) relations. 
We develop a new approach which is efficient when both $k$ and $z$ are relatively small.
The high level idea is the following.
Let $S^*\subseteq \Q(\I)$ be the optimum set of $k$ centers, and let $T^*\subseteq \I$ be the optimum set of tuple outliers in $\relclustertwo(\Q, \I,k,z)$. 
Using sampling, we partition $\I$ into two subsets $\I_1, \I_2$ such that i) $\I_1\cap \I_2=\emptyset$, ii) $\I_1\cup \I_2=\I$, and more importantly iii) $S^*\subseteq \Q(\I_1)$ and $T^*\subseteq \I_2$.
Then the relational $k$-center is executed assuming only $\I_1$. We get a set of $k$ centers $S_1\subseteq \Q(\I_1)\subseteq \Q(\I)$. Finally, we select as outliers, tuples from $\I_2$ that are joined to create join results far from $S_1$.

\vspace{-0.2em}
\paragraph{Algorithm}
We give a high level description of our algorithm. All the details can be found in Appendix~\ref{appndx:tupleoutliers}.
We repeat the following procedure for $\tau=\Theta(2^{g\cdot k+z}\log N)$ iterations. For each iteration,
initially $\I_1=\I_2=\emptyset$.
Each tuple $t\in \I$ is placed in $\I_1$ with probability $\frac{1}{2}$ and in $\I_2$ in $\frac{1}{2}$.
Notice that within $\tau$ iterations, we will find a partition such that $S^*\subseteq \Q(\I_1)$ and $T^*\subseteq \I_2$, with high probability.
We execute the oracle $\mathsf{RelCluster}(\Q,\I_1,k)$ and we get a set of $k$-centers $S_1$ and a value $r_{S_1}$ such that $\rho(S_1,\Q(\I_1))\leq r_{S_1}\leq (2+\eps)\rho^*_k(\Q(\I_1))$.
Notice that $\rho^*_k(\Q(\I_1))\leq 2\rho^*_k(\Q(\I))$ because $\Q(\I_1)\leq \Q(\I))$. Then, assume that we know the optimum $r=\hat{\rho}_{k,z}^*(\Q(\I))$. Around each point $p\in S_1$, we construct a hyper-cube $\square_p$ with center $p$ and side-length $2(r_{S_1}+r)$. Let $G=\bigcup_{p\in S_1}\square_{p}$ be the set of all constructed hyper-cubes. 
We note that all points in $G\cap \Q(\I)$ are within distance $O(r)$ from a center in $S_1$.
On the other hand, if a point $t\in\Q(\I)$ does not belong in a hyper-cube in $G$ then at least one of the tuples that are joined to construct $t$ must be an outlier, i.e., at least one of $\{\pi_{\allattr_1}(t),\ldots, \pi_{\allattr_g}(t)\}$ must be an outlier.
We identify tuples in $\Q(\I)$ that lie outside the union of $G$ by constructing 
the hyper-rectangular decomposition of the complement of $G$, denoted by $\widebar{\mathcal{M}}(G)$ (this is also equivalent to the decomposition of the complement of the arrangement of $G$~\cite{agarwal2000arrangements, halperin2017arrangements}). Since $|G|=k$, $\widebar{\mathcal{M}}(G)$ consists of $O(k^d)$ hyper-rectangles.
We go through each hyper-rectangle $\square\in \widebar{\mathcal{M}}(G)$ and we use the oracle from Lemma~\ref{lem:Rects} to get one tuple $q\in \square\cap \Q(\I)$ (if there exists one). If such $q$ is found, all tuples in $\{\pi_{\allattr_1}(q),\ldots, \pi_{\allattr_g}(q)\}$ are added in $T$. We continue with the same way until all hyper-rectangles in $\widebar{\mathcal{M}}(G)$ do not contain any point in $\Q(\I\setminus T)$. 
Notice that each time we remove at least one valid outlier tuple from $T^*\subseteq \I_2\subseteq \I$, so after at most $z$ iterations all hyper-rectangles in $\widebar{\mathcal{M}}(G)$ will be empty.
Finally, recall that we assumed knowledge of the optimum value $r = \hat{\rho}_{k,z}^*(\Q(\I))$. While this assumption does not generally hold, we can instead perform a binary search over the $L_\infty$ distances in $\Q(\I)$, as briefly discussed in Section~\ref{subsubsec:TupleOutliersOneRel} and also in~\cite{agarwal2024computing}.

In Appendix~\ref{appndx:tupleoutliers}, we show the next theorem.



\vspace{-0.2em}
\begin{theorem}
\label{thm:FPT}
        For an acyclic join query $\Q$ of $d$ attributes and $g$ tables, a database $\I$ with $|\I|=N$, and integer parameters $k, z\!\!>\!\!0$, 
        there exists an
        $(1,g,O(1))$-approximation algorithm for the $\relclustertwo(\Q,\I,k,z)$ problem, with probability at least $1-\frac{1}{N}$, that runs in $O(2^{g\cdot k+z}\cdot (z+ k^d)\cdot N\cdot \log^3 N)$ time.
\end{theorem}

\newcommand{\polylog}{\mathsf{polylog}}
\renewcommand{\fhw}{\mathsf{fhw}}
\vspace{-1em}
\subsection{Extension to Cyclic Queries}
\label{subsec:cyclic}
Similarly to~\cite{agarwal2024computing, esmailpour2024improved}, all our results for relational $k$-center clustering with outliers can be extended to cyclic join queries using the notion of \emph{fractional hypergraph width}~\cite{gottlob2014treewidth}. For a join query $\Q$, we use $\fhw(\Q)$ to denote the fractional hypergraph width of $\Q$. Intuitively $\fhw(\Q)$ shows how far $\Q$ is from being acyclic. If $\Q$ is an acyclic join query then $\fhw(Q)=1$. Using well-known and standard methodology 
an algorithm for the $\relclusterone$ or $\relclustertwo$ problem over an acyclic query $\Q$ with running time $O(\eta(k,z)\cdot N\cdot\polylog (N))$, where $\eta(k,z)$ is any real positive function over $k,z$, can be converted to an algorithm for any cyclic query with running time $O(\eta(k,z)\cdot N^{\fhw(\Q)}\cdot\polylog N)$.
\vspace{-0.6em}
\section{Conclusion}
\label{sec:conclusion}
In this work we defined the $k$-center clustering problem with set outliers. The problem has various applications in data cleaning, privacy, sensor networks etc. We gave efficient approximation algorithms in general metric spaces with arbitrary sets and in Euclidean space where sets are hyper-rectangles. Then we proposed efficient algorithms for the relational clustering problem with outliers. There are multiple open problems arising from this line of work.
An interesting direction is to study other clustering objectives such as $k$-median and $k$-means clustering in presence of set outliers. Finally, we aim to study more efficient algorithms for the relational $k$-center clustering with tuple-outliers in cases where outliers might exist in multiple tables.
\bibliographystyle{abbrv}
\bibliography{ref, acmart}

\begin{thebibliography}{10}

\bibitem{link1}
\url{https://db-engines.com/en/ranking_categories}.

\bibitem{abo2021relational}
M.~Abo-Khamis, S.~Im, B.~Moseley, K.~Pruhs, and A.~Samadian.
\newblock A relational gradient descent algorithm for support vector machine training.
\newblock In {\em Symposium on Algorithmic Principles of Computer Systems (APOCS)}, pages 100--113. SIAM, 2021.

\bibitem{abo2018database}
M.~Abo~Khamis, H.~Q. Ngo, X.~Nguyen, D.~Olteanu, and M.~Schleich.
\newblock In-database learning with sparse tensors.
\newblock In {\em Proceedings of the 37th ACM SIGMOD-SIGACT-SIGAI Symposium on Principles of Database Systems}, pages 325--340, 2018.

\bibitem{agarwal2024computing}
P.~K. Agarwal, A.~Esmailpour, X.~Hu, S.~Sintos, and J.~Yang.
\newblock Computing a well-representative summary of conjunctive query results.
\newblock {\em Proceedings of the ACM on Management of Data}, 2(5):1--27, 2024.

\bibitem{agarwal2000arrangements}
P.~K. Agarwal and M.~Sharir.
\newblock Arrangements and their applications.
\newblock In {\em Handbook of computational geometry}, pages 49--119. Elsevier, 2000.

\bibitem{agrawal2023clustering}
A.~Agrawal, T.~Inamdar, S.~Saurabh, and J.~Xue.
\newblock Clustering what matters: Optimal approximation for clustering with outliers.
\newblock {\em Journal of Artificial Intelligence Research}, 78:143--166, 2023.

\bibitem{amagata2024fair}
D.~Amagata.
\newblock Fair k-center clustering with outliers.
\newblock In {\em International Conference on Artificial Intelligence and Statistics}, pages 10--18. PMLR, 2024.

\bibitem{arenas2024towards}
M.~Arenas, T.~C. Merkl, R.~Pichler, and C.~Riveros.
\newblock Towards tractability of the diversity of query answers: Ultrametrics to the rescue.
\newblock {\em Proceedings of the ACM on Management of Data}, 2(5):1--26, 2024.

\bibitem{arora2012multiplicative}
S.~Arora, E.~Hazan, and S.~Kale.
\newblock The multiplicative weights update method: a meta-algorithm and applications.
\newblock {\em Theory of computing}, 8(1):121--164, 2012.

\bibitem{arya2000approximate}
S.~Arya and D.~M. Mount.
\newblock Approximate range searching.
\newblock {\em Computational Geometry}, 17(3-4):135--152, 2000.

\bibitem{arya1998optimal}
S.~Arya, D.~M. Mount, N.~S. Netanyahu, R.~Silverman, and A.~Y. Wu.
\newblock An optimal algorithm for approximate nearest neighbor searching fixed dimensions.
\newblock {\em Journal of the ACM (JACM)}, 45(6):891--923, 1998.

\bibitem{atserias2013size}
A.~Atserias, M.~Grohe, and D.~Marx.
\newblock Size bounds and query plans for relational joins.
\newblock {\em SIAM Journal on Computing}, 42(4):1737--1767, 2013.

\bibitem{beeri1983desirability}
C.~Beeri, R.~Fagin, D.~Maier, and M.~Yannakakis.
\newblock On the desirability of acyclic database schemes.
\newblock {\em JACM}, 30(3):479--513, 1983.

\bibitem{bentley1979decomposable}
J.~L. Bentley.
\newblock Decomposable searching problems.
\newblock {\em Information Processing Letters}, 8(5):244--251, 1979.

\bibitem{callahan1995decomposition}
P.~B. Callahan and S.~R. Kosaraju.
\newblock A decomposition of multidimensional point sets with applications to k-nearest-neighbors and n-body potential fields.
\newblock {\em Journal of the ACM (JACM)}, 42(1):67--90, 1995.

\bibitem{carmeli2023tractable}
N.~Carmeli, N.~Tziavelis, W.~Gatterbauer, B.~Kimelfeld, and M.~Riedewald.
\newblock Tractable orders for direct access to ranked answers of conjunctive queries.
\newblock {\em ACM Transactions on Database Systems}, 48(1):1--45, 2023.

\bibitem{ceccarello2019solving}
M.~Ceccarello, A.~Pietracaprina, and G.~Pucci.
\newblock Solving k-center clustering (with outliers) in mapreduce and streaming, almost as accurately as sequentially.
\newblock {\em Proceedings of the VLDB Endowment}, 12(7):766--778, 2019.

\bibitem{chan2022fully}
T.-H.~H. Chan, S.~Lattanzi, M.~Sozio, and B.~Wang.
\newblock Fully dynamic k-center clustering with outliers.
\newblock In {\em International Computing and Combinatorics Conference}, pages 150--161. Springer, 2022.

\bibitem{chan2020faster}
T.~M. Chan and Q.~He.
\newblock Faster approximation algorithms for geometric set cover.
\newblock In {\em 36th International Symposium on Computational Geometry (SoCG 2020)}, 2020.

\bibitem{charikar1999constant}
M.~Charikar, S.~Guha, {\'E}.~Tardos, and D.~B. Shmoys.
\newblock A constant-factor approximation algorithm for the k-median problem.
\newblock In {\em Proceedings of the thirty-first annual ACM symposium on Theory of computing}, pages 1--10, 1999.

\bibitem{charikar2001algorithms}
M.~Charikar, S.~Khuller, D.~M. Mount, and G.~Narasimhan.
\newblock Algorithms for facility location problems with outliers.
\newblock In {\em SODA}, volume~1, pages 642--651. Citeseer, 2001.

\bibitem{charikar2003better}
M.~Charikar, L.~O'Callaghan, and R.~Panigrahy.
\newblock Better streaming algorithms for clustering problems.
\newblock In {\em Proceedings of the thirty-fifth annual ACM symposium on Theory of computing}, pages 30--39, 2003.

\bibitem{chekuri2020fast}
C.~Chekuri, S.~Har-Peled, and K.~Quanrud.
\newblock Fast lp-based approximations for geometric packing and covering problems.
\newblock In {\em Proceedings of the Fourteenth Annual ACM-SIAM Symposium on Discrete Algorithms}, pages 1019--1038. SIAM, 2020.

\bibitem{chen2022coresets}
J.~Chen, Q.~Yang, R.~Huang, and H.~Ding.
\newblock Coresets for relational data and the applications.
\newblock {\em Advances in Neural Information Processing Systems}, 35:434--448, 2022.

\bibitem{chen2008constant}
K.~Chen.
\newblock A constant factor approximation algorithm for k-median clustering with outliers.
\newblock In {\em Proceedings of the nineteenth annual ACM-SIAM symposium on Discrete algorithms}, pages 826--835. Citeseer, 2008.

\bibitem{cheng2019nonlinear}
Z.~Cheng and N.~Koudas.
\newblock Nonlinear models over normalized data.
\newblock In {\em 2019 IEEE 35th International Conference on Data Engineering (ICDE)}, pages 1574--1577. IEEE, 2019.

\bibitem{cheng2021efficient}
Z.~Cheng, N.~Koudas, Z.~Zhang, and X.~Yu.
\newblock Efficient construction of nonlinear models over normalized data.
\newblock In {\em 2021 IEEE 37th International Conference on Data Engineering (ICDE)}, pages 1140--1151. IEEE, 2021.

\bibitem{chu2015katara}
X.~Chu, J.~Morcos, I.~F. Ilyas, M.~Ouzzani, P.~Papotti, N.~Tang, and Y.~Ye.
\newblock Katara: A data cleaning system powered by knowledge bases and crowdsourcing.
\newblock In {\em Proceedings of the 2015 ACM SIGMOD international conference on management of data}, pages 1247--1261, 2015.

\bibitem{clarkson2005improved}
K.~L. Clarkson and K.~Varadarajan.
\newblock Improved approximation algorithms for geometric set cover.
\newblock In {\em Proceedings of the twenty-first annual symposium on Computational geometry}, pages 135--141, 2005.

\bibitem{curtin2020rk}
R.~Curtin, B.~Moseley, H.~Ngo, X.~Nguyen, D.~Olteanu, and M.~Schleich.
\newblock Rk-means: Fast clustering for relational data.
\newblock In {\em International Conference on Artificial Intelligence and Statistics}, pages 2742--2752. PMLR, 2020.

\bibitem{de2000computational}
M.~De~Berg.
\newblock {\em Computational geometry: algorithms and applications}.
\newblock Springer Science \& Business Media, 2000.

\bibitem{de2023k}
M.~De~Berg, L.~Biabani, and M.~Monemizadeh.
\newblock k-center clustering with outliers in the mpc and streaming model.
\newblock In {\em 2023 IEEE International Parallel and Distributed Processing Symposium (IPDPS)}, pages 853--863. IEEE, 2023.

\bibitem{de2021k}
M.~de~Berg, M.~Monemizadeh, and Y.~Zhong.
\newblock k-center clustering with outliers in the sliding-window model.
\newblock In {\em 29th Annual European Symposium on Algorithms (ESA 2021)}, pages 13--1, 2021.

\bibitem{deep2022ranked}
S.~Deep, X.~Hu, and P.~Koutris.
\newblock Ranked enumeration of join queries with projections.
\newblock {\em Proceedings of the VLDB Endowment}, 15(5):1024--1037, 2022.

\bibitem{deep2021ranked}
S.~Deep and P.~Koutris.
\newblock Ranked enumeration of conjunctive query results.
\newblock In {\em 24th International Conference on Database Theory}, 2021.

\bibitem{dehghankar2024fair}
M.~Dehghankar, R.~Raychaudhury, S.~Sintos, and A.~Asudeh.
\newblock Fair set cover.
\newblock In {\em Proceedings of the 31st ACM SIGKDD Conference on Knowledge Discovery and Data Mining V. 1}, pages 189--200, 2025.

\bibitem{dinur2003new}
I.~Dinur, V.~Guruswami, S.~Khot, and O.~Regev.
\newblock A new multilayered pcp and the hardness of hypergraph vertex cover.
\newblock In {\em Proceedings of the thirty-fifth annual ACM symposium on Theory of computing}, pages 595--601, 2003.

\bibitem{esmailpour2024improved}
A.~Esmailpour and S.~Sintos.
\newblock Improved approximation algorithms for relational clustering.
\newblock {\em Proceedings of the ACM on Management of Data}, 2(5):1--27, 2024.

\bibitem{fagin1983degrees}
R.~Fagin.
\newblock Degrees of acyclicity for hypergraphs and relational database schemes.
\newblock {\em JACM}, 30(3):514--550, 1983.

\bibitem{feder1988optimal}
T.~Feder and D.~Greene.
\newblock Optimal algorithms for approximate clustering.
\newblock In {\em Proceedings of the twentieth annual ACM symposium on Theory of computing}, pages 434--444, 1988.

\bibitem{finkel1974quad}
R.~A. Finkel and J.~L. Bentley.
\newblock Quad trees a data structure for retrieval on composite keys.
\newblock {\em Acta informatica}, 4:1--9, 1974.

\bibitem{gonzalez1985clustering}
T.~F. Gonzalez.
\newblock Clustering to minimize the maximum intercluster distance.
\newblock {\em Theoretical computer science}, 38:293--306, 1985.

\bibitem{gottlob2014treewidth}
G.~Gottlob, G.~Greco, F.~Scarcello, et~al.
\newblock Treewidth and hypertree width.
\newblock {\em Tractability: Practical Approaches to Hard Problems}, 1:20, 2014.

\bibitem{halperin2017arrangements}
D.~Halperin and M.~Sharir.
\newblock Arrangements.
\newblock In {\em Handbook of discrete and computational geometry}, pages 723--762. Chapman and Hall/CRC, 2017.

\bibitem{har2003coresets}
S.~Har-Peled and S.~Mazumdar.
\newblock Coresets for k-means andk-median clustering and their applications.
\newblock STOC, 2003.

\bibitem{har2005fast}
S.~Har-Peled and M.~Mendel.
\newblock Fast construction of nets in low dimensional metrics, and their applications.
\newblock In {\em Proceedings of the twenty-first annual symposium on Computational geometry}, pages 150--158, 2005.

\bibitem{heinonen2001lectures}
J.~Heinonen.
\newblock {\em Lectures on analysis on metric spaces}.
\newblock Springer Science \& Business Media, 2001.

\bibitem{jiang2021faster}
S.~Jiang, Z.~Song, O.~Weinstein, and H.~Zhang.
\newblock A faster algorithm for solving general lps.
\newblock In {\em Proceedings of the 53rd Annual ACM SIGACT Symposium on Theory of Computing}, pages 823--832, 2021.

\bibitem{kara2024f}
A.~Kara, M.~Nikolic, D.~Olteanu, and H.~Zhang.
\newblock F-ivm: analytics over relational databases under updates.
\newblock {\em The VLDB Journal}, 33(4):903--929, 2024.

\bibitem{khamis2018ac}
M.~A. Khamis, H.~Q. Ngo, X.~Nguyen, D.~Olteanu, and M.~Schleich.
\newblock Ac/dc: in-database learning thunderstruck.
\newblock In {\em Proceedings of the second workshop on data management for end-to-end machine learning}, pages 1--10, 2018.

\bibitem{khot2002power}
S.~Khot.
\newblock On the power of unique 2-prover 1-round games.
\newblock In {\em Proceedings of the thiry-fourth annual ACM symposium on Theory of computing}, pages 767--775, 2002.

\bibitem{khot2008vertex}
S.~Khot and O.~Regev.
\newblock Vertex cover might be hard to approximate to within 2- $\varepsilon$.
\newblock {\em Journal of Computer and System Sciences}, 74(3):335--349, 2008.

\bibitem{klarreich2011approximately}
E.~Klarreich.
\newblock Approximately hard: The unique games conjecture.
\newblock {\em Simons foundation}, 2011.

\bibitem{krishnaswamy2018constant}
R.~Krishnaswamy, S.~Li, and S.~Sandeep.
\newblock Constant approximation for k-median and k-means with outliers via iterative rounding.
\newblock In {\em Proceedings of the 50th annual ACM SIGACT symposium on theory of computing}, pages 646--659, 2018.

\bibitem{kumar2015learning}
A.~Kumar, J.~Naughton, and J.~M. Patel.
\newblock Learning generalized linear models over normalized data.
\newblock In {\em Proceedings of the 2015 ACM SIGMOD International Conference on Management of Data}, pages 1969--1984, 2015.

\bibitem{kurkure2024faster}
Y.~Kurkure, M.~Shamo, J.~Wiseman, S.~Galhotra, and S.~Sintos.
\newblock Faster algorithms for fair max-min diversification in rd.
\newblock {\em Proceedings of the ACM on Management of Data}, 2(3):1--26, 2024.

\bibitem{malkomes2015fast}
G.~Malkomes, M.~J. Kusner, W.~Chen, K.~Q. Weinberger, and B.~Moseley.
\newblock Fast distributed k-center clustering with outliers on massive data.
\newblock {\em Advances in Neural Information Processing Systems}, 28, 2015.

\bibitem{merkl2025diversity}
T.~C. Merkl, R.~Pichler, and S.~Skritek.
\newblock Diversity of answers to conjunctive queries.
\newblock {\em Logical Methods in Computer Science}, 21, 2025.

\bibitem{moseley2021relational}
B.~Moseley, K.~Pruhs, A.~Samadian, and Y.~Wang.
\newblock Relational algorithms for k-means clustering.
\newblock In {\em International Colloquium on Automata, Languages, and Programming}, 2021.

\bibitem{rendle2013scaling}
S.~Rendle.
\newblock Scaling factorization machines to relational data.
\newblock {\em Proceedings of the VLDB Endowment}, 6(5):337--348, 2013.

\bibitem{link:relational}
Report.
\newblock Relational database management system market: Industry analysis and forecast 2021-2027: By type, deployment, end users, and region, 2022.

\bibitem{rokach2005clustering}
L.~Rokach and O.~Maimon.
\newblock Clustering methods.
\newblock {\em Data mining and knowledge discovery handbook}, pages 321--352, 2005.

\bibitem{schleich2019learning}
M.~Schleich, D.~Olteanu, M.~Abo-Khamis, H.~Q. Ngo, and X.~Nguyen.
\newblock Learning models over relational data: A brief tutorial.
\newblock In {\em Scalable Uncertainty Management: 13th International Conference, SUM 2019, Compi{\`e}gne, France, December 16--18, 2019, Proceedings 13}, pages 423--432. Springer, 2019.

\bibitem{schleich2016learning}
M.~Schleich, D.~Olteanu, and R.~Ciucanu.
\newblock Learning linear regression models over factorized joins.
\newblock In {\em Proceedings of the 2016 International Conference on Management of Data}, pages 3--18, 2016.

\bibitem{tziavelis2020optimal}
N.~Tziavelis, D.~Ajwani, W.~Gatterbauer, M.~Riedewald, and X.~Yang.
\newblock Optimal algorithms for ranked enumeration of answers to full conjunctive queries.
\newblock In {\em Proceedings of the VLDB Endowment. International Conference on Very Large Data Bases}, volume~13, page 1582. NIH Public Access, 2020.

\bibitem{tziavelis2023efficient}
N.~Tziavelis, N.~Carmeli, W.~Gatterbauer, B.~Kimelfeld, and M.~Riedewald.
\newblock Efficient computation of quantiles over joins.
\newblock In {\em Proceedings of the 42nd ACM SIGMOD-SIGACT-SIGAI Symposium on Principles of Database Systems}, pages 303--315, 2023.

\bibitem{yang2020towards}
K.~Yang, Y.~Gao, L.~Liang, B.~Yao, S.~Wen, and G.~Chen.
\newblock Towards factorized svm with gaussian kernels over normalized data.
\newblock In {\em 2020 IEEE 36th International Conference on Data Engineering (ICDE)}, pages 1453--1464. IEEE, 2020.

\bibitem{yannakakis1981algorithms}
M.~Yannakakis.
\newblock Algorithms for acyclic database schemes.
\newblock In {\em VLDB}, volume~81, pages 82--94, 1981.

\end{thebibliography}
\newpage
\appendix

\section{Missing proofs and details in Section~\ref{sec:hardness}}
\label{appndx:hardness}
\begin{proof}[Proof of Lemma~\ref{lem:hardness}]
Starting from the input $X,\mathcal{Y}$ of the $\mathsf{SC}(X,\mathcal{Y})$ problem we construct in polynomial time a set $P$ and a family of sets $\mathcal{H}$ that will serve as input to the $\prob$ problem.
For every $x_i\in X$, we place a point $p_i$ on the real line at coordinate $i$.
We define $\hat{P}=\bigcup_{x_i\in X}p_i$.
Then for an arbitrary positive integer number $k$ such that $k\leq O(n'^2)$, we add $k$ additional points $\hat{Q}=\{q_1, \ldots, q_k\}$ in $P$. 
Each point $q_j$ is placed at coordinate $2n'+j$.
Overall $P=\hat{P}\cup \hat{Q}$.
For every set $y\in \mathcal{Y}$, we create the  set $\hat{y}=\{p_i\mid x_i\in y\}$. Let $\hat{\mathcal{Y}}=\{\hat{y}\mid y\in \mathcal{Y}\}$.
Furthermore, for every $j\in[k]$ we define a set $h_j=\{q_j\}$. Let $\hat{H}=\{h_j\mid j\in [k]\}$.
Overall $\H=\hat{\mathcal{Y}}\cup \hat{H}$.
Let $n=|P|=n'+k$ and $m=|\mathcal{H}|=m'+k$,
and let $\gamma=o(m)$ be any sublinear value. We notice that for any two elements $p,q\in P$, we have $\dist(p,q)=|p-q|$, since $p,q$ are points in $\Re^1$.

    Given $X, \mathcal{Y}$ we construct in polynomial time the instance of the $\prob$ problem $P, \H$ along with a number $k$ as described above. If $\mathsf{Alg}$ is an $(1,f-\zeta,\gamma)$-approximation algorithm for the $\prob$ problem, then for any $z\in[m]$, without loss of generality, $\mathsf{Alg}$ returns a valid solution $C_z, H_z$ with $|C_z|=k$ and $|H_z|\leq (f-\zeta)z$ such that $\rho(C_z,P\setminus (\bigcup_{h\in H_z}h))\leq \gamma \cdot \rho^*_{k,z}(P,\H)$.

    We first prove some important properties. First, 
    if $\rho^*_{k,z}(P,\H)=0$ then by definition, $\rho(C_z,P\setminus (\bigcup_{h\in H_z}h))=0$.

 Furthermore, if $\rho^*_{k,z}(P,\H)=0$, then without loss of generality we can assume that $C_z=\hat{Q}$ and $H_z\subseteq \hat{\mathcal{Y}}$. Indeed, assume that $\mathsf{Alg}(P,\H,k,z)$ returns a valid solution $C_z', H_z'$ with $\rho(C_z',P\setminus (\bigcup_{h\in H_z'}h))=0$. Then it follows that for every point $p\in P$, $p\in C_z'$ or $p\in \bigcup_{h\in H_z'}h$. Assume that $C_z'\neq \hat{Q}$ or/and $H_z'\not
 \subseteq \hat{\mathcal{Y}}$. In this case, let $C_z'=Q'\cup P'$, where $Q'\subseteq \hat{Q}$ and $P'\subseteq \hat{P}$. Notice that $|Q'|+|P'|\leq k$.
 Similarly, let $H_z'=\mathcal{Y}'\cup H'$, where $\mathcal{Y}'\subseteq \hat{\mathcal{Y}}$ and $H'\subseteq \hat{H}$. Notice that $|\mathcal{Y}'|+|H'|\leq (f-\zeta)z$.
 We note that there is a $1:1$ relationship between $\hat{H}\cap H_z'$ and $\hat{P}\cap C_z'$.
 Using $C_z', H_z'$ we construct $C_z=\hat{Q}$ and $H_z\subseteq \hat{\mathcal{Y}}$ such that 
$\rho(C_z,P\setminus (\bigcup_{h\in H_z}h))=0$.
Initially, $C_z=C_z'$ and $H_z=H_z'$.
For every point $q_j\in \hat{Q}$, if $q_j\notin C_z'$, then it means that $h_j\in H_z'$. By the $1:1$ relationship, there exists a point $p_i\in \hat{P}\cap C_z'$.
Hence, we add $q_j$ in $C_z$, we remove $p_i$ from $C_z$, we remove $h_j$ from $H_z$, and we add any set $\hat{y}\in \hat{Y}$ that contains $p_i$ in $H_z$. In the end of this process notice that $C_z=\hat{Q}$, $H_z\subseteq \hat{\mathcal{Y}}$, $|C_z|=|C_z|\leq k$, $|H_z|=|H_z'|\leq (f-\zeta)z$, and every point in $\hat{Q}$ is covered by at least one set in $H_z$. Hence $C_z, H_z$ is a valid solution with $|H_z|\leq k$, $|C_z|\leq (f-\zeta)z$, and $\rho(C_z,P\setminus (\bigcup_{h\in H_z}h))=0$.

Having the two properties above, we describe our algorithm to derive a $(f-\zeta)$-approximation for the $\mathsf{SC}$-problem.
For every $z=1,2,\ldots, n'$ we run $\mathsf{Alg}(P,\H,k,z)$. We stop at the smallest $z'$ such that $\mathsf{Alg}(P,\H,k,z')$ returns $C_{z'}, H_{z'}$ where $\rho(C_{z'},P\setminus (\bigcup_{h\in H_{z'}}h))=0$. By the second property we have $C_{z'}=\hat{Q}$ and $H_{z'}\subseteq \hat{\mathcal{Y}}$. Let $Y=\{y\mid \hat{y}\in H_{z'}\}$.
Since $C_{z'}=\hat{Q}$, $\rho(C_{z'},P\setminus (\bigcup_{h\in H_{z'}}h))=0$ and $H_{z'}\subseteq \hat{\mathcal{Y}}$, it is straightforward that $Y$ is a cover for $\mathsf{SC}(X,\mathcal{Y})$.

It remains to show that $|Y|\leq (f-\zeta)\cdot \mathsf{opt}$. We first show that for any value $z\geq \mathsf{opt}$, $\rho^*_{k,z}(P,\H)=0$. Let $Y^*\subseteq \mathcal{Y}$ be the optimum solution for $\mathsf{SC}(X,\mathcal{Y})$, i.e., $Y^*$ is a cover and $|Y^*|=\mathsf{opt}$. If $z\geq \mathsf{opt}$, then consider the set $C^*=\hat{Q}$ and the family $H^*=\{\hat{y}\mid y\in Y^*\}\subseteq \H$. Notice that $|C^*|=k$, $|H^*|=\mathsf{opt}$ and every point in $P$ is either a center in $C^*$ or belongs to an outlier set in $H^*$. Hence, indeed $\rho(C^*,P\setminus(\bigcup_{h\in H^*}h))=0$ and $\rho^*_{k,z}(P,\H)=0$. Thus, $z'\leq \mathsf{opt}$. 
By the definition of $\mathsf{Alg}$, $|H_{z'}|\leq (f-\zeta)z'\leq (f-\zeta)\mathsf{opt}$. Hence, $|Y|=|H_{z'}|\leq (f-\zeta)\mathsf{opt}$. The lemma follows.
 \end{proof}

\section{Missing proofs and details from Section~\ref{subsec:GeneralDS}}
\label{appndx:generalDS}

\begin{proof}[Proof of Theorem~\ref{thm:resDS}]
Given a radius $r$, 
      the running time of our procedure is $O(nk+k^2m+\beta_1^2+(\beta_2+\beta_3)^{\expon})$ because of our analysis in Section~\ref{subsec:GeneralDS} and Section~\ref{subsec:GeneralIS}.
      In order to run a binary search over all pairwise distances we need $O(n^2\log n)$ time to sort all distances.
      
      Next we focus on the approximation factors.
      By the proof of Lemma~\ref{lem:master}, we have that $|\hat{C}|\leq 2k'$. Furthermore, $|X|=k-k'$. Hence, $|C|\leq 2k'+k-k'\leq 2k$. Similarly, from Lemma~\ref{lem:master}, $|\hat{H}|\leq 2z'$; notice that $f=1$ in this case. We have that $|H_0|=z-z'$. Hence, $|H|\leq 2z'+z-z'\leq 2z$. 
      
      Finally, we show that $\rho(C,P\setminus \bigcup_{h\in H}h)\leq 72\cdot \rho^*_{k,z}(P,\mathcal{H})$. From the proofs of Lemma~\ref{lem:master} and Lemma~\ref{helper}, we have that the ($\mathsf{LP}$\ref{lp:2}) is feasible for every $r\geq \rho^*_{k,z}(P,\mathcal{H})$ so it holds that the last radius $r'$ that our algorithm will encounter in the binary search is at most $\rho^*_{k,z}(P,\mathcal{H})$.
      From the proof of Theorem~\ref{thm:generalIS}, we have that 
    $\rho(\hat{C},P'\setminus \bigcup_{h\in \hat{H}}h)\leq 36r'\leq 
   36\rho^*_{k,z}(P,\mathcal{H})$.
    Every ball $B_\ell'$ has radius $36r'\leq 36\rho^*_{k,z}(P,\mathcal{H})$ so the point $p_\ell'$ that is added to $\hat{C}$ covers all points in the ball within distance $72\rho^*_{k,z}(P,\mathcal{H})$. Hence $\rho(C, P\setminus\bigcup_{h\in H}h)\leq 72\rho^*_{k,z}(P,\mathcal{H})$. 

    If $km<n$, then we can improve the $O(n^2\log n)$ running time to construct a sorted array of the distances to run the binary search on. After computing the $k$-center clustering using the Gonzalez algorithm in every $h_j\in \mathcal{H}$, we consider all pairwise distances among points in $\bigcup_{j\in[m]}C_j$. This leads to sorting $O(k^2m^2)$ distances in $O(k^2m^2\log(km))$ time. The approximation factor of the algorithm will increase by a constant factor. Notice that for each distance $r$, the running time of our procedure has a term $\beta_1^2$, where $\beta_1=\min\{n,km\}$. Based on the comparison of $n$ and $km$ we decide whether we sort all pairwise distances or only the pairwise distances in $\bigcup_{j\in[m]}C_j$. The result follows.
  \end{proof}

\section{Missing details and proofs in Section~\ref{subsec:GeometricIS}}
\label{appndx:geometricIS}

\paragraph{The $\mathsf{Update}(\cdot)$ procedure}

Next, we describe how we can update $\probvector$ efficiently at the beginning of each iteration.
Let $\bar{\var}$ be the solution of the oracle in the previous iteration.
Let $\delta_i=\frac{1}{\xi}(A_i\bar{\var}-b_i)=\frac{1}{k+z}(A_i\bar{\var}-b_i)=\frac{1}{k+z}(A_i\bar{\var}-1)$, where $A_i$ is the $i$-th row of $A$ ($i$-th constraint in~\eqref{neq3}).
In~\cite{arora2012multiplicative} the authors update each $\probvector[i]$ in constant time after computing $\delta_i$'s. In our case, if we attempt to calculate all $\delta_i$'s with a brute force way we would need $\Omega(n(n+m))$ time leading to a super-quadratic algorithm.
We show a faster way to calculate all $\delta_i$'s.
Our $\mathsf{Oracle}$ method sets $k+z$ variables $\bar{\var}$ to $1$.
For each $p_i\in P$ the goal is to find $A_i\bar{\var}=\sum_{p_\ell\in S_{p_i}^\eps}\bar{\var}_\ell + \sum_{\rec_j\in L_i}\bar{\var}_{n+j}=\sum_{p_\ell\in S_{p_i}^\eps, \bar{\var}_\ell> 0}\bar{\var}_\ell + \sum_{\rec_j\in L_i, \bar{\var}_{n+j}>0}\bar{\var}_{n+j}$.

We first focus on the first sum.
We modify the BBD tree $\tree$ as follows.
For each node $u\in \tree$ we define the variable $u_w=0$. For each $p_\ell$ with $\bar{\var}_\ell>0$ we start from the leaf containing $p_\ell$, and we visit the tree bottom up until we reach the root.
For each node $u$ we encounter, we update $u_w=u_w+\bar{\var}_\ell$. 
After the modification of $\tree$, for each constraint/point $p_i$ 
we query $\tree$ on $B(p_i,r)$
and we get the set of canonical nodes $\canonical(B(p_i,r))$.
We set $R_i^{(1)}\leftarrow\sum_{u\in \canonical(B(p_i,r))}u_w=\sum_{p_\ell\in S_{p_i}^\eps}\bar{\var}_\ell$.

Then, we continue with the second sum. We modify the range tree $\rangetree$ as follows. For each node $u$ in the last level tree of $\rangetree$ we define the variable $u_w=0$. For each $\rec_j\in \setrects$ such that $\bar{\var}_{n+j}>0$, we query $\rangetree$ on $\rec_j$ and we  get the set of canonical nodes $\canonicalrange(\rec_j)$. For every node $v\in \canonicalrange(\rec_j)$ we update $v_w=v_w+\bar{\var}_{n+j}$.
Then for every point $p_i\in P$ we define a set of nodes $U_i$ as follows.
For every leaf node $u$ of $\rangetree$ that contains $p_i$ we traverse the last level tree of $\rangetree$ from $u$ to the root of the tree. For every node $v$ we encounter we update $U_i=U_i\cup \{v\}$.
We set $R_i^{(2)}\leftarrow\sum_{v\in U_i}v_w=\sum_{\rec_j\in L_i}\bar{\var}_{n+j}$.

We return $\delta_i=\frac{1}{k+z}(R_i^{(1)}+R_i^{(2)}-1)$.

The correctness follows by observing that the coefficient of variable $\bar{\var}_\ell$ in the $i$-th row of $A$ is $1$ if and only if $p_\ell\in S_{p_i}^{\eps}$ and the coefficient of variable $\bar{\var}_{n+j}$ in the $i$-th row of $A$ is $1$ if and only if $\rec_j\in L_i$.

We need $O(n)$ time to compute all values $v_w$ by traversing the BBD tree $\tree$ bottom up. Then for each $p_i$ we run a query on $\tree$ in $O(\log n + \eps^{-d+1})$ time.
Overall, for the first sum we spend $O(n\log n+n\eps^{-d+1})$ time.
Next, for each $\rec_j\in \setrects$ we query $\rangetree$ on $\rec_j$ in $O(\log^d n)$ time.
A point $p_i$ lies in $O(\log^{d-1} n)$ leaf nodes in $\rangetree$ and the height of every last level subtree of $\rangetree$ is $O(\log n)$. Hence, each set $U_i$ can be computed in $O(\log^d n)$ time. Overall for the second sum we spend $O((n+m)\log^d n)$ time.
Overall,
$Q_U=O((n+m)\log^d n + n\eps^{-d+1})$.

\paragraph{The $\mathsf{Round}(\cdot)$ procedure}
The real vector $\hat{\var}$ we get satisfies ($\mathsf{LP}$\ref{lp:3}) approximately. From the MWU method (see Theorem~\ref{thm:mutli-weights}) the Constraints in $\mathcal{P}$, (Constraints~\eqref{neq1} and~\eqref{neq2}) are satisfied exactly, however the Constraints~\eqref{neq3} are satisfied approximately. In fact, it holds that
$\sum_{\rec_j\in L_i}y_j + \sum_{p_\ell\in S_{p_i}^\eps}x_\ell\geq 1-\eps$ for every $p_i\in P$.

We describe a modified version of the rounding procedure we followed in Section~\ref{subsec:GeneralIS}.

We construct a range tree $\rangetree$ over $P$. For every node $u\in \rangetree$ we initialize a list $u.list=\emptyset$.
Then, we get the set of hyper-rectangles $\rects=\{\rec_j\mid \bar{\var}_{n+j}\geq \frac{1}{2f}\}$.
For every $\rec_j\in \rects$ we compute $\canonicalrange(\rec_j)$ and for every node $u\in \canonicalrange(\rec_j)$ we add the hyper-rectangle $\rec_j$ in list $u.list$.
Next, the goal is to identify the active points $\mathsf{Act}$ in $P$, as we did in Section~\ref{subsec:GeneralIS}. Initially, $\mathsf{Act}=P$. For every point $p_i\in P$ we visit all leaf nodes of the last level subtrees that contain $p_i$ and traverse through the root of the subtree. If we find a node $u$ such that the size of $u.list$ is non-empty then we remove $p_i$ from active points.

Let $C=\emptyset$.
We construct a BBD tree $\tree$ over $\mathsf{Act}$. For every node $u$ of $\tree$ we store a representative point $u.p\in \square_u\cap \mathsf{Act}$ and a boolean variable $u.a$ which is $1$ if $u$ is active, and $0$ otherwise. Initially $u.a=1$ for every node $u\in\tree$. While the root $\mathsf{root}$ of $\tree$ is active, we repeat. Let $p_i=\mathsf{root}.p$. We add the representative $p_i$ in $C$. We define the ball $B(p_i,2r)$ and we query over $\tree$ to find $\canonical(B(p_i,2r))$.

During the query if we traverse a node $u$ such that $u.a=0$ we stop the execution through this branch of the tree. For every $u\in \canonical(B(p_i,2r))$ we set $u.a=0$. In the end we update the active nodes and representatives in $\tree$ bottom up starting from the nodes in $\canonical(B(p_i,2r))$ to the root.
The rest of the algorithm is the same as in Section~\ref{subsec:GeneralIS}.

The next Lemma follows almost verbatim from Lemma~\ref{lem:master}.
\begin{lemma}
It holds i) If $r\geq \rho^*_{k,z}(P,\setrects)$ then ($\mathsf{LP}$\ref{lp:3}) is feasible. ii) $|\rects|\leq 2fz$, and iii) If ($\mathsf{LP}$\ref{lp:3}) is feasible for a distance $r$ then $\rho(C,P\setminus (\bigcup_{\rec_j\in \rects}\rec_j))\leq 2(1+\eps)r$.
\end{lemma}
\begin{proof}
i) and ii) follow from Lemma~\ref{lem:master}, straightforwardly.

    Then we focus on iii). By definition of range tree $\rangetree$ we identify the subset of $P$ that does not belong to any hyper-rectangle $\rec_j\in \rects$. Hence, $\mathsf{Act}$ is computed correctly. Every time we add an element $p_i$ in $C$ we make inactive all nodes of $\tree$ that contains points from $\mathsf{Act}$ within distance $2r$ from $p_i$. By the approximation of the query procedure on $\tree$ we might also make inactive some nodes that contain points within distance $2(1+\eps)r$. Hence all elements in $P\setminus(\bigcup_{\rec_j\in \rects}h)j)$ are within distance $2(1+\eps)r$ from $C$.
\end{proof}

\begin{lemma}
    $|C|\leq \frac{2k}{1-2\eps}$.
\end{lemma}
\begin{proof}
    Let $p_i$ be an element that does not belong in any outlier set in $\rects$. By definition, $\sum_{\rec_j\in L_i}\bar{\var}_{n+j}\leq \frac{1}{2}$ so from the third family of inequalities in the LP, we have $\sum_{p_\ell\in S_{p_i}^\eps}\bar{\var}_{\ell}\geq \frac{1}{2}-\eps$.

We prove the lemma by an assignment argument. Every time we add a point $p$ in $C$ we assign the weight of all elements within distance $r$ from $p$.
Assume that an element $p_i$ is added in $C$ in one of the iterations of the algorithm.
The weight $w_{p_i}=\sum_{p_\ell\in S_{p_i}^\eps}\bar{\var}_\ell$ is assigned to $p_i$. We have that $w_{p_i}\geq \frac{1}{2}$ while $\sum_{i\in[n]}\bar{\var}_i=k$, and $\sum_{i\in[n], p_i\in\mathsf{Act}}\leq k$.
Furthermore, for two different elements $p_i, p_j\in C$, we have that $B(p_i,r)\cap B(p_j,r)=\emptyset$ because $\dist(p_i,p_j)>2r$ since the algorithm removes all active elements within distance at least $2r$ from a selected element in $C$.
Hence, each time that we add a new element $p_i$ in $C$ a weight $w_{p_i}\geq 1/2-\eps$ is assigned to $p$. The total amount of weight among all elements is $k$, and no fraction of weight is assigned to two different elements in $C$.
Hence $\sum_{p\in C}w_p\leq k$ and $\sum_{p\in C}w_p\geq |C|\cdot (\frac{1}{2}-\eps)$ leading to $|C|(\frac{1}{2}-\eps)\leq k\Leftrightarrow |C|\leq \frac{2k}{1-2\eps}$.
\end{proof}
If we set $\eps\leftarrow \eps/5$ we have $|C|\leq (2+\eps)k$.
The rounding is executed in $Q_R=O((n+m)\log^d n + n\cdot \eps^{-d+1})$ time.

\paragraph{Compute the set $\Gamma$}
We note that so far we assumed that $r$ is any pairwise distance. In order to find a good approximation of the optimum clustering cost, we use the notion of the Well Separated Pair Decomposition (WSPD)~\cite{callahan1995decomposition, har2005fast} briefly described in Section~\ref{sec:prelim}. Let $\Gamma$ be the sorted array of $O(n/\eps^d)$ distances from WSPD.
Any pairwise distance in $P$ can be approximated by a distance in the array $\Gamma$ within a factor $1+\eps$, hence, we might not get the optimum $\rho^*_{k,z}(P,\setrects)$ exactly. In the worst case, we might get in binary search a larger value than $\rho^*_{k,z}(P,\setrects)$ which is at most $(1+\eps)\rho^*_{k,z}(P,\setrects)$. Hence, the overall approximation with respect to the $k$-center cost increases by a $(1+\eps)$ factor.
We need $O(n\eps^{-d}\log n)$ time to compute and sort the WSPD distances, so $Q_{\Gamma}=O(n\eps^{-d}\log n)$.

\section{Missing details and proofs from Section~\ref{subsec:GeometricDS}}
\label{appndx:geometricDS}
Before we describe our algorithm we show an interesting observation for the coreset construction in Section~\ref{subsec:GeneralIS}.

If $P$ lies in a metric space with bounded doubling dimension~\cite{heinonen2001lectures} (for example $P\in \Re^d$ and $\dist$ is the Euclidean distance as in the $\gprob$ problem) then the running time is slightly better than what we reported in Section~\ref{subsec:GeneralIS}. If the doubling dimension of the metric space is $O(1)$ then we can show that $|P'|=O(\min\{n,kz\})$. Similarly to Lemma~\ref{helper}, every ball of radius $18r$ around a point in $C'$ intersects at most $z$ sets in $\mathcal{H}'$. However, notice that the points in a set $h\in \mathcal{H}'$ have minimum pairwise distance $4r$. If the doubling dimension is $O(1)$ then there are at most $O(1)$ non-intersecting balls of radius $r$ that cover all points in a ball of radius $18r$. Hence, every ball of radius $18r$ around a point in $C'$ intersects at most $z$ sets in $\mathcal{H}'$ and  covers $O(1)$ points. Hence, $|P'|=O(\min\{n,kz\})$.
We conclude with the next corollary.

 \begin{corollary}
  Given a set of $n$ elements $P$ in a metric space with bounded doubling dimension, a family of $m$ disjoint subsets $\mathcal{H}$ over $P$, and positive integer parameters $k,z$, there exists a $(2,2,O(1))$-approximation algorithm for the $\prob(P,\mathcal{H},k,z)$ that runs in $O(nk+k^2m+\beta_1^2+(\beta_2+\beta_3)^{\expon}\log \beta_1)$ time, where $\beta_1=\min\{n,km\}$, $\beta_2=\min\{n,kz\}$, and $\beta_3=\min\{m,kz\}$.
  \end{corollary}
Hence, in Section~\ref{subsec:GeometricDS} and in this Appendix, $\beta_2=\min\{n,kz\}$, instead of $\beta_2=\min\{n,km,k^2z\}$ as it was in Section~\ref{subsec:GeneralDS}.

\paragraph{Algorithm}
We use the notation from Section~\ref{subsec:GeneralDS}.
First, we apply the WSPD, described in Section~\ref{sec:prelim} to construct a set of approximate distances $\Gamma$. Then a binary search over $\Gamma$ is executed. Let $r$ be the distance we consider in one guess of the binary search.

\paragraph{Construction of $\bar{P}$ and $\bar{\setrects}$}
We construct a range tree $\rangetree$ over $P$.
For every hyper-rectangle $\rec_j\in \setrects$, we run a reporting query on $\rangetree$ and we get all points $P_j=P\cap \rec_j$.
We run the Feder and Greene algorithm ~\cite{feder1988optimal} (efficient geometric algorithm for $k$-center clustering) on each set $P_j$ independently. For each $j\in[m]$, we get the set of centers $C_j$.
If $\rho(C_j,P_j)> 2r$ then we remove all points $P_j$ and the hyper-rectangle $\rec_j$ because it must be an outlier as we had in Section~\ref{subsec:GeneralDS}.
Then for every remaining $j\in [m]$, we construct a BBD tree $\tree_j$ over $C_j$. We use the BBD tree to remove all necessary points from $C_j$ so that no two points in $C_j$ have distance less than $2r$. For every $c\in C_j$, we define the ball $B(c,2r)$ and we run the query $\tree_j(B(c,2r))$. It returns the set $\canonical(B(c,2r))$. For every point $c'\in C_j\cap (\bigcup_{u\in \canonical(B(c,2r))})$, we remove $c'$ from $C_j$. We continue this operation for every remaining point in $C_j$.

In this way we define $\bar{P}, \bar{\setrects}, \bar{z}, \beta_1$ similarly to Section~\ref{subsec:GeneralDS}. For each remaining $j\in [m]$, let $\bar{P}_j=\bar{P}\cap \rec_j$.
From the additional $(1+\eps)$-approximation factor we get from the BBD tree, we have that $\rho(C_j,\rec_j)\leq 4(1+\eps)r$ and $\rho_{k,\bar{z}}^*(\bar{P}, \bar{\setrects})\leq 5(1+\eps)r$.

So far we spent $O(n\log^{d-1}n)$ time to construct $\rangetree$, and $O(n+m\log^{d-1}n)$ time to construct all sets $P_j$. For every $P_j$, Feder and Greene algorithm runs in $O(|P_j|\log k)$ time, hence all sets $C_j$ are constructed in $O(n\log k)$ time. The construction of all BBD trees $\tree_j$ take $O(n\log n)$ time, and for each query ball, the reporting query in the BBD tree takes $O(\eps^{-d+1}+\log n +\mathsf{OUT})$, where $\mathsf{OUT}$ is the size of the output. Hence, we remove all necessary points in $O(km(\eps^{-d+1}+\log n))$ time. Overall the running time so far is $O((n+m)\log^{d-1}(n)+km(\eps^{-d+1}+\log n))$.

\paragraph{Construction of $P'$ and $\setrects'$}
The main operator we should implement efficiently is:
For every point $p_i\in \bar{P}$, we should check whether the set  $B_i=B(p_i,18(1+\eps)r))\cap \bar{P}$ contains points from more than $\bar{z}$ hyper-rectangles.
If such a point is found then we remove all points in $B(p_i,36(1+\eps)r))\cap \bar{P}$ and we continue in the next iteration.
We propose a technical algorithm using a BBD tree over $\bar{P}$ that runs in $(\beta_1(\eps^{-d+1}+\log\beta_1))$ time.

In order to execute this part efficiently, we construct a BBD tree $\tree$ over $\bar{P}$ and a copy of the BBD tree denoted by $\tree'$.
If not mentioned explicitly, all operations and canonical nodes are getting from tree $\tree$.
For every node $u$ of $\tree$ we store a set of indexes $u.s$ which is initially empty.
Furthermore, for every index $j$ stored in $u.s$, we also store a counter $u.\mathsf{count}(j)$.
For every $j\in[m]$ and for every $p\in \bar{P}_j$ we 
we query $\tree$ on $B(p,18(1+\eps)r)$ and we get the set of canonical nodes $\canonical(B(p,18(1+\eps)r))$.
For every $u\in\canonical(B(p,18(1+\eps)r))$
we check whether $j\in u.s$. If no, then we add $j$ in $u.s$ and we set $u.\mathsf{count}(j)=1$. If yes, then we set $u.\mathsf{count}(j)\leftarrow u.\mathsf{count}(j)+1$. After repeating this step for all $j\in[j]$ and all points in $\bar{P}$, we go again through every $j\in[m]$ and $p\in \bar{P}_j$ and we run the following procedure. We get again $\canonical(B(p,18(1+\eps)r))$. For every $u\in \canonical(B(p,18(1+\eps)r))$ we keep $j$ in $u.s$ if and only if there is no node $v$ from the root of $\tree$ to $u$ such that $j\in v.s$. Hence, we go through every node $v$ in the path from $u$ to root and check whether $j\in v.s$.
First, assume that $j\in v.s$, where $v$ is the node closer to root that contains the index $j$ in its set.
We remove $j$ from $u.s$, and we set $v.\mathsf{count}(j)\leftarrow v.\mathsf{count}(j)+1$.
If we do not find such a node $v$ then we do not make any change.
Notice that in the end of this process for any path from any leaf node of $\tree$ to the root of $\tree$ an index $j$ might belong to at most one set $v.s$, where $v$ is a node from the leaf to the roof of $\tree$.

Then, we go through each leaf node of $\tree$. For a leaf node $u$, let $p_i\in\bar{P}_j$ be the point stored in the leaf node. We visit every node $v$ from $u$ to the root of $\tree$ and we count the total number of indexes stored in the sets of the nodes: If $\mathsf{path}_{u\rightarrow \mathsf{root}}$ is the set of nodes from $u$ to the root of $\tree$, then we get $\mathsf{count}_{p_i}=\sum_{v\in \mathsf{path}_{u\rightarrow \mathsf{root}}}|v.s|$. If $\mathsf{count}_{p_i}>\bar{z}$ then 
using $\tree'$ we get $B_i'=\canonical'(B(p,36(1+\eps)r))\cap \bar{P}$, where $\canonical'(B(p,36(1+\eps)r))$ contains the canonical (active) nodes in tree $\tree'$ of the query ball $B(p,36(1+\eps)r)$.
We implicitly remove $B_i'$ from $\tree'$ by making the nodes in $\canonical'(B(p,36(1+\eps)r))$ inactive.
We add the set $B_i'$ in the family of sets $X$, as defined in Section~\ref{subsec:GeneralDS}.

It remains to update the stored information in the nodes of the BBD tree $\tree$ after removing $B_i'$.
For every point $p\in B_i'$, we get the canonical nodes $\canonical(B(p,18(1+\eps)r))$. Without loss of generality, assume that $p\in \bar{P}_j$. For every $u\in \canonical(B(p,18(1+\eps)r))$ we get the ancestor node $v$ (it is unique) such that $j\in v.s$. Notice that it might be the case that $v=u$. We update $v.\mathsf{count}(j)\leftarrow v.\mathsf{count}(j)-1$. If $v.\mathsf{count}(j)=0$ then we remove $j$ from $v.s$.

Finally, we continue the same method through all leaf nodes of $\tree$ skipping nodes that contain points that are already removed from $\bar{P}$.

Let $P', \setrects', k', z'$ be defined similarly to Section~\ref{subsec:GeneralDS}.
Using the proof of Lemma~\ref{helper} along with the additional $(1+\eps)$-approximation factor of the BBD tree, the remark of Section~\ref{subsec:GeneralDS}, and the additional $(1+\eps)$-approximation factor we get from the binary search on $\Gamma$, we conclude that $|P'|=O(\min\{n,kz\})$, $|\setrects'|=O(\min\{m,kz\})$, and $\rho^*_{k',z'}(P',\setrects')\leq 18(1+\eps)^4r$.
If we set $\eps\leftarrow \eps/6$ we have $\rho^*_{k',z'}(P',\setrects')\leq 18(1+\eps)r$.

The additional running time to construct $P', \setrects'$ from $\bar{P}$ and $\bar{\setrects}$ is computed as follows. The BBD trees $\tree, \tree'$ are constructed in $O(\beta_1\log\beta_1)$ time. For every $p_i$ we compute $\canonical(B(p_i,18(1+\eps)r))$ in $O(\eps^{-d+1}+\log\beta_1)$ time. Over all points in $\bar{P}$, we spend $O(\beta_1(\eps^{-d+1}+\log\beta_1))$ time. This is also the running time we need to check whether there are ancestor nodes with the same index $j$ stored in the nodes' sets. The running time of all operations in $\tree'$ is also bounded by $O(\beta_1(\eps^{-d+1}+\log\beta_1))$ because for every $p$, $\canonical'(B(p,36(1+\eps)r))$ is computed in $O(\eps^{-d+1}+\log\beta_1)$ time.
Overall, given $\bar{P}, \bar{\setrects}$, we construct $P', \setrects'$ in $O(\beta_1(\eps^{-d+1}+\log\beta_1))$ time.

\paragraph{Solve LP on $P', \setrects'$ using MWU} Finally, the MWU algorithm from Section~\ref{subsec:GeometricIS} is executed on $P', \setrects'$ with parameters $k',z'$.

\ifarxiv
\section{$k$-center clustering with result-outliers}
\label{subsec:RelClustResOutliers}
We first show the connection of the $\relclusterone$ problem  with the standard (geometric) $k$-center clustering with outliers problems in the standard computational setting.
Notice that $\Q(\I)$ is a set of tuples/points in $\Re^d$.
Hence, the $k$-Center clustering with result-outliers can be mapped to the standard problem of $k$-center clustering with outliers in a set $\Q(\I)$ of $|\Q(\I)|$ points in $\Re^d$.
While there are known approximation algorithms~\cite{charikar2001algorithms, charikar2003better} for the $k$-center problem with outliers in the standard computational setting, any (super-)linear time algorithm with respect to the input size cannot be run in the relational setting efficiently because $|\Q(\I)|$ might be huge compared to the size of the database.
Instead, we implement in the relational setting the sampling procedure from~\cite{charikar2003better}, that selects a small number of samples and then run the more expensive algorithm from~\cite{charikar2001algorithms} on the set of samples. In~\cite{charikar2003better} the authors show that running the algorithm from~\cite{charikar2001algorithms} on the samples leads to a good overall approximation for the $k$-center clustering with outliers.

For simplicity, we first briefly describe the algorithm in the standard computational setting considering a set $P$ of $n$ points in $\Re^d$. In this case, we assume that $z=\delta\cdot n$, for $\delta\in(0,1)$.
First, sample $\tau=\Theta(\frac{k\log n}{\eps^2\delta})$ points from $P$ uniformly at random. Let $Q\subseteq P$ be the set of samples. Then, run the algorithm for $k$-center clustering with outliers from~\cite{charikar2001algorithms} on $Q$. Let $C=\emptyset$ be the set of centers we will select.
We run a binary search on the pairwise distances of $Q$. For every distance $r$, we run the following for at most $k$ iterations. Let $p\in Q$ be the point such that $|B(p,r)\cap Q|$ is maximized. The point $p$ is added in $S$ and all points from $P$ in $B(p,3r)\cap Q$ are removed from $Q$.
If after $k$ iterations the remaining number of points from $Q$ is at most $(1+\eps)\delta\tau$ then we set $T=Q$ (the remaining points from $Q$), we store the solution $C, T$ and we continue the binary search for smaller values of $r$. Otherwise, if $|Q|>(1+\eps)\delta\tau$, then we skip the current solution and we continue the binary search for larger values of $r$. In the end, we return the last solution $C, T$ computed by the binary search. This algorithm satisfies that $|T|\leq (1+\eps)^2\delta n\Leftrightarrow |T|\leq (1+\eps)^2z$ and $\rho(C,P\setminus T)\leq 3 \cdot\opt$, with probability at least $1-\frac{1}{n}$ where $\opt$ is $k$-center cost of the optimum solution. The running time of this algorithm as implemented in~\cite{charikar2001algorithms} is $O(n+\tau^2)$.

Next, we show how to implement the algorithm above in the relational setting. Furthermore, we describe a faster (approximate) implementation of the algorithm in~\cite{charikar2001algorithms} in $\Re^d$ using BBD trees, in the standard computational setting.
We assume that $|\Q(\I)|=\omega(N)$. Otherwise, we simply compute $\Q(\I)$ using Yannakakis algorithm~\cite{yannakakis1981algorithms} and  execute our fast implementation of~\cite{charikar2001algorithms} on $\Q(\I)$, as shown below.


\paragraph{Algorithm}
We run Yannakakis algorithm~\cite{yannakakis1981algorithms} to get $|\Q(\I)|$.
Let $z=\delta |\Q(\I)|$.
Using the oracle from Lemma~\ref{lem:Rects}, we get a set $Q\subseteq \Q(\I)$ with $\tau=\Theta(\frac{k\log |\Q(\I)|}{\eps^2\delta})$ samples from $\Q(\I)$, for a parameter $\eps\in(0,1)$. 
If the algorithm from~\cite{charikar2001algorithms} was executed on $Q$ we would have an additional runtime of $O(\tau^2)$ time. We show a faster implementation.
Let $\tree$ be a BBD tree constructed on $Q$. For every node $u$ of $\tree$ we store $u.a=1$ denoting that the node $u$ is active and $u.c=|\square_u\cap Q|$.
Let $S=\emptyset$.
We construct a WSPD as described in Section~\ref{sec:prelim}. Let $\Gamma$ be the sorted array of $O(\tau\eps^{-d})$ distances. The binary search is executed on $\Gamma$.
Let $r$ be a distance in $\Gamma$.
For every $p_i\in Q$, such that the set of nodes in $\tree$ from the root to the leaf node storing $p_i$ does not contain any inactive node,
we query $\tree$ with the $B(p_i,r)$ and it returns the set of canonical nodes $\canonical(B(p_i,r))$.
If we visit a node $u$ with $u.a=0$ then we stop the execution through this branch of the tree.
We set $c_i=\sum_{u\in \canonical(B(p_i,r))}u.c$. Let $p_{i^*}$ be the point such that $c_{i^*}$ is maximized. We add $p_{i^*}$ in $S$ and we search $\tree$ with the ball $B(p_{i^*},3r)$ (without considering branches of the tree with inactive nodes). Let $\canonical(B(p_{i^*},3r))$ be the returned set of canonical nodes.  For every node $u\in \canonical(B(p_{i^*},3r))$, we set $u.a=0$ and then we update all nodes from $\canonical(B(p_{i^*},3r))$ to the root of $\tree$ bottom up as follows. Let $u$ be a node and $v, w$ be its children. Initially, $u$ is the parent node of the deepest node in $\canonical(B(p_{i^*},3r))$. If $v.a=w.a=0$ then we set $u.a=0$. If $v.a=0$ and $w.a=1$ then $u.c=w.c$. If $v.a=1$ and $w.a=0$ then $u.c=v.c$. Finally, if $v.a=w.a=1$ then $u.c=v.c+w.c$. We continue the update with the parent node of $u$. After $k$ iterations we check whether $u^{(1)}.c>(1+\eps)\delta\tau$, where $u^{(1)}$ is the root of $\tree$. If yes, then we continue the binary search with larger values of $r$. If no, then we add in $T$ every point $p\in Q$ such that there is no inactive node from $u^{(1)}$ to the leaf node storing $p$.
In the end we return the last stored $S, T$.

\paragraph{Proof of correctness and runtime}
\begin{lemma}
    $\rho(S,\Q(\I)\setminus T)\leq 3(1+\eps)^2\rho^*_{k,z}(\Q(\I))$ and $|T|\leq (1+\eps)^2\delta|\Q(\I)|$ with probability at least $1-\frac{1}{N}$
\end{lemma}
\begin{proof}
    The correctness proof of our algorithm follows from the properties of the BBD tree. Every time that a node $u$ becomes inactive the points in $\square_u$ do not contribute in the counts of the active nodes, i.e., the update operation of the counts is correct. In each iteration, it finds a point $p_{i^*}$ with the maximum number of points in a region that contains $B(p_{i^*},r)$ and might contain some parts of $B(p_{i^*},(1+\eps)r)\setminus B(p_{i^*},r)$. Furthermore, every time that $p_{i^*}$ is added in $S$ all points in $Q$ within distance $3r$ and some within distance $3(1+\eps)r$ are explicitly removed from $Q$, i.e., the canonical nodes that contain these points become inactive. The analysis is the same as in~\cite{charikar2001algorithms} ($3$-approximation), however we have a multiplicative $(1+\eps)$ error due to the BBD tree approximation. Furthermore, because of the WSPD we might not get $r=\opt$. Instead, we guarantee that we will find at least a distance at most $(1+\eps)\opt$. Hence the overall approximation is $3(1+\eps)^2$ with probability at least $1-\frac{1}{N^2}$.

    The second part of the lemma $|T|\leq (1+\eps)^2\delta|\Q(\I)|$ follows directly from~\cite{charikar2003better}.
\end{proof}
Notice that $3(1+\eps/6)^2\leq 3+\eps$ and $(1+\eps/6)^2\leq (1+\eps)$.

\begin{lemma}
The running time of the algorithm is $O(N\log N +\eps^{-d}\tau\log \tau + k(\tau\log\tau + \tau\eps^{-d+1})\log\tau)$.
\end{lemma}
\begin{proof}
    Using Lemma~\ref{lem:Rects} we can sample $\tau$ points from $\Q(\I)$ in $O(N\log N+\tau\log N)$ time. The BBD tree $\tree$ is constructed in $O(\tau\log \tau)$ time and the set $\Gamma$ in $O(\eps^{-d}\tau\log \tau)$ time. The binary search is executed for at most $O(\log (\eps^{-d}\tau))=O(\log \tau)$ times. The next steps are executed at most $k$ times. We find all points that are not contained in any inactive node in $O(\tau\log\tau)$ time. For each such point $p_i$ we compute $\canonical(B(p_i,r))$ in $O(\log \tau +\eps^{-d+1})$ time. Furthermore the set $\canonical(B(p_{i^*},3r))$ and the update procedure is executed in $O(\log \tau +\eps^{-d+1})$ time, from the definition of the BBD tree. 
\end{proof}

We conclude with the next theorem.
\begin{theorem}
    \label{thm:main1}
    For an acyclic join query $\Q$ of $d$ attributes, a database $\I$ of input size $N$ with $|\Q(\I)|=\omega(N)$, integer positive parameters $k, z$, and a small constant $\eps\in (0,1)$, 
    there exists an
    $(1,1+\eps,3+\eps)$-approximation for the $\relclusterone(\Q,\I,k,z)$ problem with probability at least $1-\frac{1}{N}$ that runs in $O\left(N\log N + \frac{k^2}{\delta}\log^3 N\right)$ time, where $\delta=\frac{z}{|\Q(\I)|}$. If $|\Q(\I)|=O(N)$, then the running time is bounded by $O(kN\log^2 N)$.
\end{theorem}
\fi

\section{Missing details and proofs in Section~\ref{subsec:RelClustRupleOutliers}}
\label{appndx:tupleoutliers}

\subsection{Correctness and runtime  in Section~\ref{subsubsec:TupleOutliersOneRel}}

The correctness of our approach follows from the correctness of the binary search from~\cite{agarwal2024computing}, the correctness of our algorithm in Section~\ref{subsec:GeometricDS}, and the correctness of the relational $k$-center clustering in Lemma~\ref{lem:relkCenter}.
We have $|S|\leq (2+\eps)k$ and $|T|\leq 2z$.
The binary search on $L_\infty$ norm might not find the optimum cost $\hat{\rho}^{*}_{k,z,1}(\Q(\I))$ but a number $r$ such that $\hat{\rho}^{*}_{k,z,1}(\Q(\I))\leq r\leq (1+\eps)\hat{\rho}^{*}_{k,z,1}(\Q(\I))$, similarly to the WSPD in Section~\ref{subsec:GeometricDS}. 
Thus, our algorithm has the same approximation factor as in Theorem~\ref{thm:coresetgeomfinaldisjoin}.

The overall time to compute the pairwise distances to check in the binary search takes $O(N\log^2 N)$ time. For each pairwise distance $r$, and for every tuple $t\in R_1$, we run the algorithm from Lemma~\ref{lem:relkCenter} in $O(k^2N)$ time. Hence, in total, the running time to run all instances of the relational $k$-center clustering algorithm over all distance guesses is $O(k^2\cdot N^2 \cdot \log N)$. Using the algorithm from Theorem~\ref{thm:geomfinal} (for $f=1$), we get that the total running time of our algorithm is $O\left((k^2N^2+kN(k+z)\log^{d+1}N)\log N\right)$.
Using the algorithm from Theorem~\ref{thm:coresetgeomfinaldisjoin}, the running time slightly improves to $O((k^2N^2+ kN\log^{d-1}(N) +(k+z)(kz+\beta_3)\log^{d+2}(kz))\cdot \log N)$, where $\beta_3=\min\{N,kz\}$.
For simplicity, and in most real applications, we can assume that the first term dominates all the other terms in the running time.

\subsection{Proof of Theorem~\ref{thm:FPT}}
We first describe the complete algorithm from Section~\ref{subsubsec:TupleOutliersMultRels}.

\paragraph{Algorithm}
We repeat the following procedure for $\tau=\Theta(2^{g\cdot k+z}\log N)$ iterations. For iteration $\ell$,
initially $\I_1=\I_2=\emptyset$.
Each tuple $t\in \I$ is placed in $\I_1$ with probability $\frac{1}{2}$ and in $\I_2$ in $\frac{1}{2}$.
We execute the oracle $\mathsf{RelCluster}(\Q,\I_1,k)$ and we get a set of $k$-centers $S_1$ and a value $r_{S_1}$ such that $\rho(S_1,\Q(\I_1))\leq r_{S_1}\leq (2+\eps)\rho^*_k(\Q(\I_1))$.
Then, we execute the binary search on the $L_\infty$ norm of the pairwise distances in $\Q(\I)$ as briefly described in Section~\ref{subsubsec:TupleOutliersOneRel} (also in~\cite{agarwal2024computing}). For every $r$ in the binary search we set $T'=\emptyset$ and $\hat{r}=r_{S_1}+\sqrt{d}r$. For every $p\in S_1$ we define a hyper-cube $\square_p$ with center $p$ and side length $2\hat{r}$. In other words, $\square_p$ is the ball of center $p$ and radius $\hat{r}$ in $L_\infty$ norm. Let $G=\bigcup_{p\in S_1}\square_p$ and let $\widebar{\mathcal{M}}(G)$ be the hyper-rectangular decomposition of the complement of $G$~\cite{agarwal2000arrangements, halperin2017arrangements}. For each hyper-cube $\square\in \widebar{\mathcal{M}}(G)$ we check whether $\square\cap \Q(\I\setminus T')$ is empty using the oracle in Lemma~\ref{lem:Rects}. If it is empty we continue with the next hyper-cube in $\widebar{\mathcal{M}}(G)$. If it is not empty, then we use again the oracle from Lemma~\ref{lem:Rects} to get any tuple $q\in \square\cap \Q(\I)$. We check whether $q\in \Q(\I_2)$. If no, then we break this execution of the algorithm and we continue the binary search with larger values of $r$. If indeed $q\in \Q(\I_2)$ then we add in $T'$ all tuples (one from every relation) whose join construct $q$, i.e., $T'=T'\cup (\bigcup_{j\in[g]}\pi_{\allattr_j}(g))$. 
We continue getting the tuples $q$ until there is no other remaining tuple over all hyper-cubes in $\widebar{\mathcal{M}}$  or until we visited $z$ tuples in total.
In the end, we check whether there exists a square $\square \in \widebar{\mathcal{M}}(G)$ such that $\square\cap \Q(\I\setminus T')\neq\emptyset$. If yes, then we break the execution of the algorithm and we continue the binary search with larger values of $r$. If not, then we say that it was a valid iteration of the binary search; we store the set of tuple outliers $T'$ and the cost $r'=\hat{r}$.

After the end of the execution of the binary search, let $S_\ell=S_1$, $T_\ell=T'$ and $r_\ell=r'$ where $S_1$, $T'$, $r'$ are the set of centers, tuple outliers and radius, respectively, in the last valid iteration of the binary search.
After repeating the above procedure for $\tau$ iterations we set $T=T_{\ell^*}$, where $\ell^*=\argmin_{\ell\in [\tau]}r_\ell$. For every hyper-cube $\square\in \bigcup_{p\in S_{\ell^*}}\square_p$ we get one representative point $s_{\square}\in \square\cap \Q(\I\setminus T)$ running the oracle in Lemma~\ref{lem:Rects}. We set $S=\{s_\square\mid \square\in \bigcup_{p\in S_{\ell^*}}\square_p\}$ and we return $S, T$.

We prove Theorem~\ref{thm:FPT} showing the next two lemmas.

\begin{lemma}
\label{lem:last1}
    It holds that $|S|\leq k$, $|T|\leq g\cdot z$ and $\rho(S,\Q(\I\setminus T))\leq O(1)\cdot \hat{\rho}^*_{k,z}(\Q(\I))$, with probability at least $1-\frac{1}{N}$.
\end{lemma}
\begin{proof}
    By definition for every center in $S_{\ell^*}$ we add a point in $S$. Recall that $S_{\ell^*}$ is the set of $k$ centers returned by the relational $k$-center clustering algorithm from Lemma~\ref{lem:relkCenter}, so indeed $|S|\leq k$.
    Next, we note that in every iteration of the algorithm we visit at most $z$ points from $\Q(\I)\cap \widebar{\mathcal{M}}(G)$. For each such points we add at most $g$ tuples from $\I$ in $T'$. Thus for every $\ell\in [\tau]$, $|T_{\ell}|\leq g\cdot z$, and hence $|T|\leq g\cdot z$.

    Next, we show the approximation guarantee that holds with high probability. The probability that all tuples from $\I$ that are joined to generate $S^*$ are included in $\I_1$ and all tuples in $T^*$ are included in $\I_2$ is $\prod_{t\in T}\frac{1}{2}\prod_{s\in \mathsf{Tuples}(S^*)}\frac{1}{2}$, where $\mathsf{Tuples}(S^*)\subseteq \I$ is the set of tuples such that are needed to joined to generate all points in $S^*$. By definition $|S^*|\leq k$ and there are $g$ tables in $\allrel$ so $|\mathsf{Tuples}(S^*)|\leq g\cdot k$. Hence the probability that all tuples from $\I$ that are joined to generate $S^*$ are included in $\I_1$ and all tuples in $T^*$ are included in $\I_2$ is bounded by $\frac{1}{2^{g\cdot k+z}}$. Since the algorithm is repeated for $\Theta(2^{g\cdot k+z}\log N)$ iterations, there will be at least one iteration where $S^*\subseteq \Q(\I_1)$ and $T^*\subseteq \I_2$ with probability at least $1-\frac{1}{N}$. Without loss of generality, assume that this is the $\ell$-th iteration.

Let $r^*=\rho^*(S^*,\Q(\I\setminus T^*))$.
By the definition of the binary search on $L_\infty$ distances, it will guess a value $r$ such that $r\leq r^*\leq \sqrt{d}r$.
We show that this iteration of the binary search indeed computes $T_\ell$ (a valid step of the binary search).
Let $t\in \Q(\I)\cap \widebar{\mathcal{M}}(G)$. We show that at least one tuple that is joined to construct $t$ belongs in $T^*$.
We have
$\dist(t,S^*)>r_{S_1}+\sqrt{d}r - r_{S_1}\geq r^*$.
Hence $t\in\Q(\I_2)$
and indeed every tuple in $\Q(\I)\cap \widebar{\mathcal{M}}(G)$ contains at least one tuple outlier from $T^*$. Every time we visit a tuple $t\in \Q(\I\setminus T')\cap \widebar{\mathcal{M}}(G)$ we remove all the tuples in $\I$ that are joined to construct it. Thus after visiting $z$ tuples we have removed at least $z$ tuple outliers, and by definition no other tuple should remain in $\Q(\I\setminus T')\cap \widebar{\mathcal{M}}(G)$. We conclude that $T_\ell=T'$ is set in our algorithm.

From our discussion above the set $T_{\ell^*}$ is well-defined.
By definition, all points in $\Q(\I\setminus T_{\ell^*})$ are within (Euclidean) distance at most $\sqrt{d}r_{\ell^*}$ from the points in $S_{\ell^*}$.
In the end, the set $S$ is constructed by adding one  point in $\square\cap \Q(\I\setminus T_{\ell^*})$ for every $\square\in \bigcup_{p\in S_{\ell^*}}\square_p$. So all points in $\Q(\I\setminus T_{\ell^*})$ are actually within distance $2\sqrt{d}r_{\ell^*}$ from the points in $S$.
By definition $2\sqrt{d}r_{\ell^*}\leq 2\sqrt{d} r_\ell\leq 2\sqrt{d}r'\leq 2\sqrt{d}(r_{S_1}+\sqrt{d}r)$.
It holds that $r\leq r^*$, and 
$\rho(S_1,\Q(\I_1))\leq r_{S_1}\leq (2+\eps)\rho^*_k(\Q(\I_1))\leq 2(2+\eps)r^*\leq 6r^*$. The  inequality $(2+\eps)\rho^*_k(\Q(\I_1))\leq 2(2+\eps)r^*$ holds because $\Q(\I_1)\subseteq \Q(\I\setminus T^*)$.
Hence, for every point $q\in \Q(\I\setminus T_{\ell^*})$, it holds that $\dist(q,S)\leq 2\sqrt{d}r_{\ell^*}\leq 2(d+6\sqrt{d})r^*=2(d+6\sqrt{d})\rho^*(S^*,\Q(\I\setminus T^*))$. 
We conclude that $\rho(S,\Q(\I\setminus T))\leq 2(d+6\sqrt{d})\hat{\rho}^*_{k,z}(\Q(\I))$, with probability at least $1-\frac{1}{N}$. The lemma follows because $2(d+6\sqrt{d})=O(1)$.

\end{proof}

\begin{lemma}
\label{lem:last2}
    The algorithm runs in $O(2^{g\cdot k+z}\cdot (z+ k^d)\cdot N\log^3 N)$ time.
\end{lemma}
\begin{proof}
    The algorithm is repeated for $O(2^{g\cdot k+z}\log N)$ iterations. For each iteration we run the relational $k$-center clustering from Lemma~\ref{lem:relkCenter} on $\Q,\I_1$ in $O(k^2N)$ time. Then the binary search takes $O(N\log^2 N)$ time to generate all $O(\log N)$ distances that our algorithm checks~\cite{agarwal2024computing}. For every step of the binary search, we construct the complement of the arrangement $\widebar{\mathcal{M}}(G)$ in $O(k^d)$ time~\cite{agarwal2000arrangements}. For every cell $\square\in \widebar{\mathcal{M}}(G)$, we run the oracle in Lemma~\ref{lem:Rects} in $O(N\log N)$ time to count $|\square\cap \Q(\I)|$ or get a point from $\square\cap \Q(\I)$. In total, we allow to visit at most $z$ points over  $\{\square\mid\square\in \widebar{\mathcal{M}}(G)\}\cap \Q(\I)$ Furthermore, notice that it is straightforward to check whether $t\in \Q(\I_2)$ in $O(g\log N)=O(\log N)$ time. The running time follows.
\end{proof}

\end{document}